\documentclass{lmcs}
\pdfoutput=1

\usepackage{lastpage}
\lmcsdoi{18}{1}{17}
\lmcsheading{}{\pageref{LastPage}}{}{}%
{Mar.~15,~2021}{Jan.~19,~2022}{}

\keywords{type theory, equality checking, proof assistant}

\usepackage[T1]{fontenc}
\usepackage[utf8]{inputenc}  % to allow unicode in source

\usepackage[english]{babel}

\usepackage{amsfonts}
\usepackage{amsmath,amsthm} % basic ams math environments and symbols
\usepackage{amssymb}    % ams symbols and theorems

\usepackage{mathtools} % for \mathrlap
\usepackage{upgreek} % for upright greek letters
\usepackage{scalerel}

\usepackage{mdframed}

\usepackage{dutchcal} % for script letters

\usepackage{xcolor} % for colors all around and a rainbow on top
\definecolor{rulenameColor}{rgb}{0.5,0.5,0.5}

\usepackage[hidelinks]{hyperref}

\usepackage[capitalise,noabbrev,nameinlink]{cleveref} % load it after hyperref

\makeatletter
\newcommand{\addresseshere}{%
  \enddoc@text\let\enddoc@text\relax
}
\makeatother

\usepackage[shortlabels]{enumitem} % for customizing enumerated lists some more
\usepackage{mathpartir} % for inference rules

\usepackage{listings} % package for code
\usepackage{lstandromeda2} % Andromeda2 code coloring

\makeatletter

\let\c@thm\relax

\let\c@prop\relax

\let\c@lem\relax

\let\c@cor\relax

\let\c@defi\relax

\let\c@exa\relax
\makeatletter

{
\theoremstyle{plain}
\newtheorem{thm}{Theorem}[section]
\newtheorem{prop}[thm]{Proposition}
\newtheorem{lem}[thm]{Lemma}
\newtheorem{cor}[thm]{Corollary}
}

{
\theoremstyle{definition}
\newtheorem{defi}[thm]{Definition}
\newtheorem{exa}[thm]{Example}
}

\crefname{thm}{Theorem}{Theorems}
\crefname{cor}{Corollary}{Corollaries}
\crefname{lem}{Lemma}{Lemmas}
\crefname{prop}{Proposition}{Propositions}
\crefname{defi}{Definition}{Definitions}
\crefname{exa}{Example}{Examples}

\newcommand{\defemph}[1]{\emph{#1}} % defined term

\newcommand{\dummy}{{\star}} % dummy value
\newcommand{\mto}{{\mapsto}} % short \mapsto (can you find a shorter arrow?)
\newcommand{\finmap}[1]{\langle #1 \rangle} % finite map

\newcommand{\Inst}[4]{#1 \in \mathrm{Inst}(#2, #3, #4)} % \Inst{I}{\Xi}{\Theta}{\Gamma} is an instantiation of \Xi over \Theta and \Gamma
\newcommand{\upto}[2]{#1_{(#2)}} % restriction of an instantiation
\newcommand{\act}[1]{#1_*} % action of an instantiation
\newcommand{\upact}[2]{#1_{(#2)*}} % combination of restriction and action

\newcommand{\residue}[2]{#1/#2} % Equational residue of mv-extension #1 w.r.t instantiation #2
\newcommand{\residueInst}[1]{#1^r} % Equational residue of an instantiation #1

\newcommand{\basic}[3]{\mathbb{I}(#1,#2,#3)} % basic instantiation of metavariable extension #1 at metavariable #2 with argument #3
\newcommand{\basicCod}[3]{\mathbb{E}(#1,#2,#3)} % codomain of basic instantiation of metavariable extension #1 at metavariable #2 with argument #3
\newcommand{\basicComp}[2][]{\mathbb{J}_{#1}(#2)} % composite of basic instantiations for the list of metavariables in #1
\newcommand{\basicCompCod}[2][]{\mathbb{F}_{#1}(#2)} % codomain of composite of basic instantiations for the list of metavariables in #1

\newcommand{\bnfis}{\mathrel{\;{:}{:}{=}\ }}
\newcommand{\bnfor}{\mathrel{\;\big|\ \ }}

\newenvironment{ruleframe}{\begin{mdframed}[linecolor=gray,linewidth=0pt]}{\end{mdframed}}

\newcommand{\synclass}[1]{{\scaleobj{0.9}{\mathsf{#1}}}}

\newcommand{\Ty}{\synclass{Ty}} % the type syntactic class
\newcommand{\Tm}{\synclass{Tm}} % the term syntactic class
\newcommand{\EqTy}{\synclass{EqTy}} % the type equality syntactic class
\newcommand{\EqTm}{\synclass{EqTm}} % the term equality syntactic class

\newcommand{\sym}[1]{\mathsf{#1}} % a symbol
\newcommand{\symS}{\sym{S}} % generic symbol
\newcommand{\symM}{\sym{M}} % generic meta-variable

\newcommand{\arity}[1]{\mathsf{ar}(#1)} % the arity of a symbol
\newcommand{\rawRule}[2]{#1 \Longrightarrow #2} % the raw rule
\newcommand{\natty}[2]{\tau_{#1}(#2)} % the natural type

\newcommand{\set}[1]{\{#1\}} % a set of things

\newcommand{\abstr}[1]{\{#1\}} % abstraction
\newcommand{\B}{\mathcal{b}} % generic boundary thesis
\newcommand{\BB}{\mathcal{B}} % abstracted boundary thesis

\newcommand{\J}{\mathcal{j}} % generic judgement thesis
\newcommand{\JJ}{\mathcal{J}} % abstracted judgement thesis

\newcommand{\var}[1]{\mathsf{#1}} % a font for free variables/atoms

\newcommand{\isType}[1]{#1\;\mathsf{type}}
\newcommand{\isCtx}[1]{#1\;\mathsf{vctx}}
\newcommand{\isExt}[1]{#1\;\mathsf{mctx}}

\newcommand{\natur}{\types^{\natural}} % natural-for-variable judgement

\newcommand{\of}{{:}} % a : A with less spacing around the colon
\newcommand{\asset}[1]{\{\kern-0.237em| #1 |\kern-0.237em\}} % assumption set

\newcommand{\fv}[1]{\mathsf{fv}(#1)} % the deep free variables of an expression
\newcommand{\mv}[1]{\mathsf{mv}(#1)} % the metavariables of an expression
\newcommand{\bv}[1]{\mathsf{bv}(#1)} % the bound variables in an expression

\newcommand{\bye}{\;\mathsf{by}\;} % A \equiv B by \alpha

\newcommand{\objmv}[1]{|#1|_\mathsf{obj}} % the object metavariables of metavariable extension #1

\newcommand{\types}{\vdash} % another name for entailment
\newcommand{\plug}[2]{#1{\setlength{\fboxrule}{0.5pt}\setlength{\fboxsep}{1pt}\fbox{\vphantom{$#1$}$#2$}}} % fill the head of a boundary

\newcommand{\emptyCtx}{[]} % empty context
\newcommand{\emptyExt}{[]} % empty metavariable extension

\newcommand{\genapp}[1]{\widehat{#1}} % generic application

\newcommand{\rulename}[1]{\textnormal{\textsc{#1}}}
\newcommand{\rref}[1]{\hyperref[#1]{\rulename{#1}}} % hyperlink a rule

\newcommand{\inCase}[1]{\smallskip\par\noindent\textit{Case \rref{#1}}:}

\newcommand{\inCaseText}[1]{\smallskip\par\noindent\textit{Case #1}:}

\newcommand{\inferenceRule}[3]{\inferrule*[lab={\label{#1}\rulename{\color{rulenameColor}#1}}]{#2}{#3}}

\newcommand{\cmp}{\sim}
\newcommand{\cmpN}{\sim_\mathrm{n}}
\newcommand{\cmpE}{\sim_\mathrm{e}}

\newcommand{\principal}[1]{\wp(#1)} % The principal arguments
\newcommand{\compute}{\triangleright} % normalization

\newcommand{\computeC}{\triangleright_\mathrm{c}} % Apply a computation rule
\newcommand{\computeP}{\triangleright_\mathrm{p}} % Normalize the principal arguments

\newcommand{\matches}[4]{#1 \types #2 \triangleright #3 \leadsto #4}  % in meta variable extension #1, match pattern #2 with expression #3 to yield instantiation #4
\newcommand{\notmatches}[3]{#1 \types #2 \triangleright #3 \, {\not\leadsto}}  % matching fails

\title[An extensible equality checking algorithm]{An extensible equality checking algorithm\texorpdfstring{\\}{} for dependent type theories}

\author[A.~Bauer]{Andrej~Bauer}
\author[A.~Petković Komel]{Anja~Petković Komel}
\address{Faculty of Mathematics and Physics, University of Ljubljana}
\email{Andrej.Bauer@andrej.com, Anja.Petkovic@fmf.uni-lj.si}

\begin{document}

\begin{abstract}
  We present a general and user-extensible equality checking algorithm that is applicable to a large class of type theories. The algorithm has a type-directed phase for applying extensionality rules and a normalization phase based on computation rules, where both kinds of rules are defined using the type-theoretic concept of object-invertible rules. We also give sufficient syntactic criteria for recognizing such rules, as well as a simple pattern-matching algorithm for applying them. A third component of the algorithm is a suitable notion of principal arguments, which determines a notion of normal form. By varying these, we obtain known notions, such as weak head-normal and strong normal forms.
  We prove that our algorithm is sound.
  We implemented it in the Andromeda~2 proof assistant, which supports user-definable type theories. The user need only provide the equality rules they wish to use, which the algorithm automatically classifies as computation or extensionality rules, and select appropriate principal arguments.
\end{abstract}

\maketitle

% chktex-file 46
% chktex-file 45
% chktex-file 1
% chktex-file 26
% chktex-file 24
\section{Introduction}
\label{sec:introduction}

Equality checking algorithms are essential components of proof assistants based on type theories~\cite{coq-site,agda-site,lean,coq-correct,gilbert19:_defin_k,agda-conversion}. They free users from the burden of proving scores of mostly trivial judgemental equalities, and provide computation-by-normalization engines. Some systems~\cite{dedukti-site,sprinklesAgda} go further by allowing user extensions to the built-in equality checkers.

The situation is less pleasant in a proof assistant that supports arbitrary user-definable theories, such as Andromeda~2~\cite{andromeda-site,andromeda-1}, where in general no equality checking algorithm may be available.
Nevertheless, the proof assistant should still provide support for equality checking that is easy to use and works well in the common, well-behaved cases. For this purpose we have developed and implemented a sound and extensible equality checking algorithm for user-definable type theories.

The generality of type theories supported by Andromeda~2 presents a significant challenge in devising a useful equality checking algorithm. Many commonly used ideas and notions that one encounters in specific type theories do not apply anymore: not every rule can be classified either as an introduction or an elimination form, not every equation as either a $\beta$- or an $\eta$-rule, all terms must be fully annotated with types to ensure soundness, there may be no reasonable notion of normal form, or neutral form, etc. And of course, the user may easily define a theory whose equality checking is undecidable. In order to do better than just exhaustive proof search, some compromises must therefore be made and design decisions taken:
\begin{enumerate}
\item We work in the fully general setting of standard type theories (\cref{sec:finit-type-theor}).
\item We prefer ease of experimentation at the expense of possible non-termination or unpredictable behavior.
\item At the same time, soundness of the algorithm is paramount: any equation verified by it must be derivable in the theory at hand.
\item The algorithm should work well on well-behaved theories, and especially those seen in practice.
\end{enumerate}
The most prominent design goals missing from the above list are completeness and performance. The former cannot be achieved in full generality, as there are type theories with undecidable equality checking. We have expended enough energy looking for acceptably general sufficient conditions guaranteeing completeness to state with confidence that this task is best left for another occasion. Regarding performance, we freely admit that equality checking in Andromeda~2 is nowhere near the efficiency of established proof assistants. For this we blame not only the immaturity of the implementation, but also the generality of the situation, which simply demands that a price be paid in exchange for soundness.
We console ourselves with the fact that our equality checker achieves soundness and complete user-extensibility at the same time.

\subsection*{Contributions}
\label{sec:contributions}

We present a \emph{general equality checking algorithm} that is applicable to a large class of type theories, the \emph{standard type theories} of~\cite{bauer:_finit} (\cref{sec:finit-type-theor}).
The algorithm (\cref{sec:type-directed-equality-checking}) is fashioned after equality checking algorithms~~\cite{Harper-Stone,Abel-Scherer} that have a type-directed phase for applying extensionality rules (inter-derivable with $\eta$-rules), intertwined with a normalization phase based on computation rules ($\beta$-rules).
For the usual kinds of type theories (simply typed $\lambda$-calculus, Martin-Löf type theory, System F), the algorithm behaves like the well-known standard equality checkers.
We prove that our algorithm is \emph{sound} (\cref{sec:sound-algor}).

We define a general notion of \emph{computation} and \emph{extensionality rules} (\cref{sec:extens-comp-rules}),
using the type-theoretic concept of an \emph{object-invertible rule} (\cref{sec:invertible-rules}). We also provide sufficient \emph{syntactic criteria} for recognizing such rules, together with a simple pattern-matching algorithm for applying them.
A third component of the algorithm is a suitable notion of normal form, which guarantees correct execution of normalization and coherent interaction of both phases of the algorithm.
In our setting, normal forms are determined by a selection of \emph{principal arguments} (\cref{sec:principal-and-normalization}). By varying these, we obtain known notions, such as weak head-normal and strong normal forms.

We \emph{implemented} the algorithm in Andromeda~2 (\cref{sec:implementation}).
The user need only provide the equality rules they wish to use, which the algorithm automatically classifies either as computation or extensionality rules, rejects those that are of neither kind, and selects appropriate principal arguments.

Those readers who prefer to see examples before the formal development, may first take a peek at~\cref{sec:examples-code}, where we show how our work allows one to implement extensional type theory, and use the reflection rule to derive computation rules which are only available in their propositional form in intensional type theory.

\subsection*{Acknowledgments}

We thank Philipp G.\ Haselwarter for his support and discussions through which he generously shared ideas that helped get this work completed.
This material is based upon work supported by the U.S.~Air Force Office of Scientific Research under award number FA9550-17-1-0326, grant number 12595060, and award number FA9550-21-1-0024.

%%% Local Variables:
%%% mode: latex
%%% TeX-master: "equality-checking"
%%% End:

% chktex-file 46
% chktex-file 45
% chktex-file 1
% chktex-file 26
% chktex-file 24
\section{Finitary type theories}
\label{sec:finit-type-theor}

We shall work with a variant of general dependent type theories~\cite{bauer20:gtt}, namely \emph{finitary type theories}, as described in~\cite{bauer:_finit} and implemented in Andromeda~2. We give here only an overview of the syntax of such theories and refer the reader to~\cite{bauer:_finit} for a complete exposition.

\subsection{Deductive systems}
\label{sec:deductive-system}

We first recall the general notion of a deductive system.
A (finitary) \defemph{closure rule} on a carrier set~$S$ is a pair $([p_1, \ldots, p_n], q)$ where $p_1, \ldots, p_n, q \in S$. The elements $p_1, \ldots, p_n$ are the \defemph{premises} and $q$ is the \defemph{conclusion} of the rule. A rule may be displayed as
\begin{equation*}
  \infer{p_1 \\ \cdots \\ p_n}{q}.
\end{equation*}
A \defemph{deductive system} on a set~$S$ is a family $C$ of closure rules on~$S$.
We say that $T \subseteq S$ is \defemph{deductively closed} for~$C$ when the following holds: for every rule $C_i = ([p_1, \ldots, p_n], q)$, if $\{p_1, \ldots, p_n\} \subseteq T$ then $q \in T$.
A \defemph{derivation} with \defemph{conclusion~$q \in S$} is a well-founded tree whose root is labeled by an index~$i$ of a closure rule $C_i = ([p_1, \ldots, p_n], q)$, and whose subtrees are derivations with conclusions~$p_1, \ldots, p_n$. We say that $q \in S$ is \defemph{derivable} if there exists a derivation with conclusion~$q$. The derivable elements of~$S$ form precisely the least deductively closed subset.

All deductive systems that we shall consider will have as their carriers the set of hypothetical judgements and boundaries, as described in \cref{sec:judg-bound}.

\subsection{Signatures and arities}
\label{sec:signatures-arities}

In a finitary type theory there are four \defemph{judgement forms}:
\begin{itemize}
\item ``$\isType{A}$'' asserting that $A$ is a type,
\item ``$t : A$'' asserting that $t$ is a term of type $A$,
\item ``$A \equiv B \bye \dummy_\Ty$'' asserting that types $A$ and $B$ are equal, and
\item ``$s \equiv t : A \bye \dummy_\Tm$'' asserting that terms $s$ and $t$ are equal at type~$A$.
\end{itemize}
We indicate these with tokens $\Ty$, $\Tm$, $\EqTy$ and $\EqTm$ respectively.
To each token there also corresponds a syntactic class. Expressions of class~$\Ty$ are the \emph{type expressions}, and those of class~$\Tm$ are the \emph{term expressions}.
These are formed using \emph{(primitive) symbols} and \emph{metavariables}, see \cref{sec:expressions}, each of which has an associated arity, as explained below. The symbols should be thought of as the primitive type and term formers, while the metavariables shall be used to refer to the premises of a rule, and as pattern variables in the equality checking algorithm.
The only expressions of syntactic classes $\EqTy$ and $\EqTm$ are the dummy expressions~$\dummy_\Ty$ and $\dummy_\Tm$, which we both write as~$\dummy$ when no confusion can arise. These are formality, to be used where one would normally record a proof term witnessing a premise, but the premise is a judgemental equality, which is proof irrelevant.

The \defemph{symbol arity $(c, [(c_1, n_1), \ldots, (c_k, n_k)])$} of a symbol~$\symS$ tells us that
\begin{enumerate}
\item the syntactic class of expressions built with $\symS$ is $c \in \set{\Ty, \Tm}$,
\item $\symS$ accepts $k \in \mathbb{N}$ arguments,
\item the $i$-th argument has syntactic class~$c_i \in \set{\Ty, \Tm, \EqTy, \EqTm}$ and binds~$n_i \in \mathbb{N}$ variables.
\end{enumerate}

\begin{exa}
  The arity of a type constant such as $\mathsf{bool}$ is $(\Ty, [])$, the arity of a binary term operation such as~$+$ is $(\Tm, [(\Tm, 0), (\Tm, 0)])$, and the arity of a quantifier such as the dependent product~$\Uppi$ is $(\Ty, [(\Ty, 0), (\Ty, 1)])$ because it is a type former taking two type arguments, with the second one binding one variable.
\end{exa}

The \defemph{metavariable arity} of a metavariable~$\symM$ is a pair $(c, n)$, where the \defemph{syntactic class} $c \in \set{\Ty, \Tm, \EqTy, \EqTm}$ indicates whether~$\symM$ is respectively a type, term, type equality, or term equality metavariable, and $n \in \mathbb{N}$ is the number of term arguments it accepts. The metavariables of syntactic classes $\Ty$ and $\Tm$ are the \defemph{object metavariables}, and they participate in formation of expressions, while those of syntactic classes $\EqTy$ and $\EqTm$ are the \defemph{equality metavariables}, and are used to refer to equational premises.

The information about symbol arities is collected in a \defemph{signature}, which maps each symbol to its arity. When discussing syntax, it is understood that such a signature has been given, even if we do not mention it explicitly.

\subsection{Expressions}
\label{sec:expressions}

The syntax of finitary type theories is summarized in the top part of
\cref{fig:syntax-general-type-theories}. There are three kinds: type expressions, term
expressions, and arguments.

\begin{figure}[htbp]
  \centering
  \small
  \begin{ruleframe}
  \begin{align*}
  \text{Type expression}\ A, B
  \bnfis& \symS(e_1, \ldots, e_n)   &&\text{type symbol application}\\
  \bnfor& \symM(t_1, \ldots, t_n)   &&\text{type metavariable application}
  \\
  \text{Term expression}\ s, t
  \bnfis& \var a                         &&\text{free variable}\\
  \bnfor& x                              &&\text{bound variable}\\
  \bnfor& \symS(e_1, \ldots, e_n)   &&\text{term symbol application}\\
  \bnfor& \symM(t_1, \ldots, t_n)   &&\text{term metavariable application}
  \\
  \text{Argument}\ e
  \bnfis& A           &&\text{type argument} \\
  \bnfor& t           &&\text{term argument} \\
  \bnfor& \dummy_\Ty  &&\text{dummy type equality argument} \\
  \bnfor& \dummy_\Tm  &&\text{dummy term equality argument} \\
  \bnfor& \abstr x e  &&\text{abstraction ($x$ bound)}
  \\[1ex]
  \text{Judgement thesis}\ \J
  \bnfis& \isType{A}                    && \text{$A$ is a type} \\
  \bnfor& t : A                         && \text{$t$ has type $T$} \\
  \bnfor& A \equiv B \bye \dummy_\Ty     && \text{$A$ and $B$ are equal types} \\
  \bnfor& s \equiv t : A \bye \dummy_\Tm && \text{$s$ and $t$ are equal terms at $A$}
  \\
  \text{Abstracted judgement:}\ \JJ
  \bnfis& \J                   &&\text{judgement thesis} \\
  \bnfor& \abstr{x \of A} \JJ  &&\text{abstracted judgement ($x$ bound)}
  \\[1ex]
  \text{Boundary thesis}\ \B
  \bnfis& \isType{\Box}            &&\text{a type}\\
  \bnfor& \Box :  A                &&\text{a term of type $A$}\\
  \bnfor& A \equiv B \bye \Box      &&\text{type equation boundary}\\
  \bnfor& s \equiv t : B \bye \Box  &&\text{term equation boundary}
  \\
  \text{Abstracted boundary}\ \BB
  \bnfis& \B                   &&\text{boundary thesis} \\
  \bnfor& \abstr{x \of A} \BB  &&\text{abstracted boundary ($x$ bound)}
  \\[1ex]
  \text{Variable context}\ \Gamma
  \bnfis& \mathrlap{[\var{a}_1 \of A_1, \ldots, \var{a}_n \of A_n]}
  \\[1ex]
  \text{Metavariable context}\ \Theta
  \bnfis& \mathrlap{[\symM_1 \of \BB_1, \ldots, \symM_n \of \BB_n]}
  \\[1ex]
  \text{Hypothetical judgement}\
  \phantom{\bnfis}& \Theta; \Gamma \types \JJ \\
  \text{Hypothetical boundary}\
  \phantom{\bnfis}& \Theta; \Gamma \types \BB
  \end{align*}
  \end{ruleframe}
  \caption{The syntax of expressions, boundaries and judgements.}
  \label{fig:syntax-general-type-theories}
\end{figure}

A \defemph{type expression} is an application $\symS(e_1, \ldots, e_n)$ of a \defemph{primitive symbol~$\symS$} to arguments, or an application $\symM(t_1, \ldots, t_n)$ of a \defemph{metavariable~$\symM$} to terms. We write $\symS$ and $\symM$ instead of $\symS()$ and $\symM()$.

A \defemph{term expression} is a variable, an application of a primitive symbol to arguments, or an application of a metavariable to terms. We strictly separate free variables $\var{a}, \var{b}, \var{c}, \ldots$ from the bound ones $x, y, z, \ldots$, a choice fashioned after the locally nameless syntax~\cite{mckinna93:_pure_type_system_formal,chargueraud12:_local_namel_repres}, a common implementation technique in which free variables are represented as names and the bound ones as de Bruijn indices.

An \defemph{argument} is a type expression, a term expression, a dummy argument~$\star_\Ty$ or $\star_\Tm$, or an abstracted argument $\abstr{x} e$ binding~$x$ in~$e$.
Note that we take abstraction to be a basic syntactic operation. For instance, we do not construe a $\lambda$-abstraction as a variable-binding construct $\lambda x \of A \,.\, t$, but rather an application $\uplambda(A, \abstr{x} t)$ of the primitive symbol~$\uplambda$ to two separate arguments $A$ and $\abstr{x} t$.
We may abbreviate an iterated abstraction $\abstr{x_1} \cdots \abstr{x_n} e$ as $\abstr{\vec{x}} e$, and similarly use the vector notation elsewhere when appropriate.
We permit $\vec{x}$ to be empty, in which case $\abstr{\vec{x}} e$ is just~$e$.
To an argument we assign the metavariable arity
\begin{equation*}
  \arity{\abstr{x_1} \cdots \abstr{x_n} e} = (c, n),
\end{equation*}
where $c \in \set{\Ty, \Tm, \EqTy, \EqTm}$ is the syntactic class of the non-abstracted argument~$e$.

For an expression to be syntactically valid, all bound variables must be bound by abstractions, and all symbol and metavariable applications respect their arities.
That is, if the arity of $\symS$ is $(c, [(c_1, n_1), \ldots, (c_k, n_k)])$ then it must be applied to~$k$ arguments $e_1, \ldots, e_k$ with $\arity{e_i} = (c_i, n_i)$, and the expression $\symS(e_1, \ldots, e_k)$ has syntactic class~$c$.
Similarly, an object metavariable $\symM$ of arity $(c, n)$ must be applied to~$n$ term
expressions to yield an expression of syntactic class~$c$.

We write $e[t/x]$ for capture-avoiding \defemph{substitution} of~$t$ for $x$ in $e$, and $e[t_1/x_1, \ldots, t_n/x_n]$ or $e[\vec{t}/\vec{x}]$ for simultaneous substitution of $t_1, \ldots, t_n$ for $x_1, \ldots, x_n$. Expressions which only differ in the choice of names of bound variables are considered syntactically identical (alternatively, we could use de Bruijn indices for bound variables).

Given an expression~$e$, let $\mv{e}$ and $\fv{e}$ be the sets of metavariables and free variables occurring in~$e$, respectively.
A \defemph{renaming} of an expression $e$ is an injective map $\rho$ with domain $\mv{e} \cup \fv{e}$ that takes metavariables to metavariables and free variables to free variables. The renaming acts on~$e$ to yield an expression $\rho_* e$ by replacing each occurrence of a metavariable~$\symM$ and a free variable~$\var{a}$ with $\rho(\symM)$ and $\rho(\var{a})$, respectively. We similarly define renamings of metavariable and variable contexts, judgements, and boundaries, which are defined below.

\subsection{Judgements and boundaries}
\label{sec:judg-bound}

We next discuss the syntax of judgements and boundaries, see the bottom part of \cref{fig:syntax-general-type-theories}.
To each of the judgement forms corresponds a \defemph{judgement thesis}:
\begin{itemize}
\item ``$\isType{A}$'' asserts that $A$ is a type,
\item ``$t : A$'' that $t$ is a term of type~$A$,
\item ``$A \equiv B \bye \dummy_\Ty$'' that types $A$ and $B$ are equal, and
\item ``$s \equiv t : A \bye \dummy_\Tm$'' that terms $s$ and $t$ of type $A$ are equal.
\end{itemize}
The latter two have ``$\bye \dummy$'' attached so that all boundaries can be filled with a head, as we shall explain shortly. We normally write just ``$A \equiv B$'' and ``$s \equiv t : A$''.

A \defemph{boundary} is a fundamental notion of type theory, although perhaps less familiar.
Whereas a judgement is an assertion, a boundary is a \emph{goal} to be accomplished:
\begin{itemize}
\item ``$\isType{\Box}$'' asks that a type be constructed,
\item ``$\Box : A$'' that the type $A$ be inhabited, and
\item ``$A \equiv B \bye \Box$'' and ``$s \equiv t : A \bye \Box$'' that equations be proved.
\end{itemize}
An \defemph{abstracted judgement} has the form $\abstr{x \of A}\; \JJ$, where $A$ is a type expression and $\JJ$ is a (possibly abstracted) judgement. The variable~$x$ is bound in $\JJ$ but not in~$A$. As before, we write $\abstr{\vec{x} \of \vec{A}} \; \J$ for an iterated abstraction
$
  \abstr{x_1 \of A_1} \cdots \abstr{x_n \of A_n} \; \J.
$
Similarly, an \defemph{abstracted boundary} has the form
$
  \abstr{x_1 \of A_1} \cdots \abstr{x_n \of A_n} \; \B
$,
where~$\B$ is a \defemph{boundary thesis}, i.e., it takes one of the four (non-abstracted) boundary forms.

To an abstracted boundary we assign a metavariable arity by
\begin{align*}
  \arity{\abstr{x_1 \of A_1} \cdots \abstr{x_n \of A_n} \BB} &= (c, n)
\end{align*}
where $c \in \set{\Ty, \Tm, \EqTy, \EqTm}$ is the syntactic class of the non-abstracted boundary~$\B$.

The placeholder $\Box$ in a boundary $\BB$ may be filled with an argument~$e$, called the \defemph{head}, to give a judgement $\plug{\BB}{e}$, as follows:
\begin{align*}
  \plug{(\isType{\Box})}{A} &= (\isType{A}), \\
  \plug{(\Box : A)}{t} &= (t : A), \\
  \plug{(A \equiv B \bye \Box)}{\dummy} &= (A \equiv B \bye \dummy), \\
  \plug{(s \equiv t : A \bye \Box)}{\dummy} &= (s \equiv t : A \bye \dummy), \\
  \plug{(\abstr{x \of A} \BB)}{\abstr{x} e} &= (\abstr{x \of A} \plug{\BB}{e}).
\end{align*}
We also define the operation $\plug{\BB}{e \equiv e'}$ which turns an object boundary~$\BB$ into an equation:
\begin{align*}
  \plug{(\isType{\Box})}{A \equiv B} &= (A \equiv B \bye \dummy), \\
  \plug{(\Box : A)}{s \equiv t} &= (s \equiv t : A \bye \dummy), \\
  \plug{(\abstr{x \of A} \BB)}{\abstr{x} e \equiv \abstr{x} e'} &= (\abstr{x \of A} \plug{\BB}{e \equiv e'}).
\end{align*}

\begin{exa}
  If the symbols $\sym{A}$  and  $\sym{Id}$ have arities
  \begin{equation*}
    (\Ty, []),
    \quad\text{and}\quad
    (\Ty, [(\Ty, 0), (\Tm, 0), (\Tm, 0)]),
  \end{equation*}
  respectively, then the boundaries
  \begin{equation*}
    \abstr{x \of \sym{A}} \abstr{y \of \sym{A}} \; \Box : \sym{Id}(\sym{A}, x, y)
    \qquad\text{and}\qquad
    \abstr{x \of \sym{A}} \abstr{y \of \sym{A}} \; x \equiv y : \sym{A} \bye \Box
  \end{equation*}
  may be filled with heads $\abstr{x} \abstr{y} x$ and $\abstr{x} \abstr{y} \dummy$ to yield abstracted judgements
  \begin{equation*}
    \abstr{x \of \sym{A}} \abstr{y \of \sym{A}} \;  x : \sym{Id}(\sym{A}, x, y)
    \qquad\text{and}\qquad
    \abstr{x \of \sym{A}} \abstr{y \of \sym{A}} \; x \equiv y : \sym{A} \bye \dummy_\Tm.
  \end{equation*}
\end{exa}

Information about available metavariables is collected in a \defemph{metavariable context}, which is a finite list $\Theta = [\symM_1 \of \BB_1, \ldots, \symM_n \of \BB_n]$, also construed as a map, assigning to each metavariable $\symM_i$ a boundary $\BB_i$.
In \cref{sec:raw-rules}, $\Theta$ will provide typing of metavariable and premises of an inference rule, while the level of raw syntax it just determines metavariable arities. That is, $\Theta$ assigns the metavariable arity $\arity{\BB_i}$ to~$\symM_i$.

A metavariable context $\Theta = [\symM_1 \of \BB_1, \ldots, \symM_n \of \BB_n]$ may be \emph{restricted} to a metavariable context $\upto{\Theta}{i} = [\symM_1 \of \BB_1, \ldots, \symM_{i-1} \of \BB_{i-1}]$.

The metavariable context~$\Theta$ is syntactically well formed when each~$\BB_i$ is a
syntactically well-formed boundary over~$\Sigma$ and~$\upto{\Theta}{i}$. In addition each $\BB_i$ must be closed, i.e., contain no free variables.

We also define the set of the object metavariables of~$\Theta$ to be
\begin{equation*}
  \objmv{\Theta} = \set{\symM_i \mid \text{$\BB_i$ is an object boundary}}.
\end{equation*}

A \defemph{variable context}~$\Gamma = [\var{a}_1 \of A_1, \ldots, \var{a}_n \of A_n]$ over a metavariable context~$\Theta$ is a finite list of pairs written as $\var{a}_i \of A_i$. It is considered syntactically valid when the variables $\var{a}_1, \ldots, \var{a}_n$ are all distinct, and for each~$i$ the type expression $A_i$ is valid with respect to the signature and the metavariable arities assigned by~$\Theta$, and the free variables occurring in~$A_i$ are among $\var{a}_1, \ldots, \var{a}_{i-1}$. A variable context $\Gamma$ yields a finite map, also denoted $\Gamma$, defined by $\Gamma(\var{a}_i) = A_i$. The \emph{domain} of~$\Gamma$ is the set $|\Gamma| = \{\var{a}_1, \ldots, \var{a}_n\}$.

A \defemph{context} is a pair $\Theta; \Gamma$ consisting of a metavariable context~$\Theta$ and a variable context~$\Gamma$ over $\Theta$. A syntactic entity is considered syntactically valid over a signature and a context $\Theta; \Gamma$ when all symbol and metavariable applications respect the assigned arities, the free variables are among~$|\Gamma|$, and all bound variables are properly abstracted. It goes without saying that we always require all syntactic entities to be valid in this sense.

A \defemph{(hypothetical) judgement} has the form
\begin{equation*}
  \Theta; \Gamma \types \JJ,
\end{equation*}
where $\Theta; \Gamma$ is a context and~$\JJ$ is an abstracted judgement over $\Theta; \Gamma$.
In a hypothetical judgement
\begin{equation*}
  \Theta ; \var{a}_1 \of A_1, \ldots, \var{a}_n \of A_n
  \types
  \abstr{x_1 \of B_1} \cdots \abstr{x_m \of B_m} \J
\end{equation*}
the hypotheses are split between the variable context $\var{a}_1 \of A_1, \ldots, \var{a}_n \of A_n$ on the left of~$\types$, and the abstraction $\abstr{x_1 \of B_1} \cdots \abstr{x_m \of B_m}$ on the right. The former lists the \emph{global} hypotheses that interact with other judgements, and the latter the hypotheses that are \emph{local} to the judgement.
In our experience such a separation is quite useful, because it explicitly marks the part of the context that is abstracted when a variable-binding symbol is applied to its arguments.

A \defemph{(hypothetical) boundary} is formed in the same fashion, as
\begin{equation*}
  \Theta; \Gamma \types \BB.
\end{equation*}
We read it as asserting that $\BB$ is a well-typed boundary over the context~$\Theta; \Gamma$.

\subsection{Instantiations}
\label{sec:raw-instantiations}

Let us spell out how how to instantiate metavariables with arguments.
An \defemph{instantiation} of a metavariable context $\Xi = [\symM_1 \of \BB_1, \ldots, \symM_n \of \BB_n]$ over a context $\Theta; \Gamma$ is a list representing a map
\begin{equation*}
  \finmap{\symM_1 \mto e_1, \ldots, \symM_n \mto e_n},
\end{equation*}
where $e_i$'s are arguments over $\Theta; \Gamma$ such that $\arity{\BB_i} = \arity{e_i}$. We sometimes write $\Inst{I}{\Xi}{\Theta}{\Gamma}$ when~$I$ is such an instantiation.

For $k \leq n$, define the \defemph{restriction}
\begin{equation*}
  \upto{I}{k} = \finmap{\symM_1 \mto e_1, \ldots, \symM_{k-1} \mto e_{k-1}}.
\end{equation*}
We sometimes write $\upto{I}{\symM}$ to indicate the initial segment up to the given metavariable $\symM \in |I|$.
We use the same notation for initial segments of sequences in general, e.g., if $\vec{x} = (x_1, \ldots, x_n)$ then $\upto{\vec{x}}{k} = (x_1, \ldots, x_{k-1})$.

An \defemph{instantiation $I$ acts} on an expression $e$ to give an expression $\act{I} e$ by:
\begin{gather*}
  \act{I} \var{a} = \var{a},
  \qquad
  \act{I} x = x,
  \qquad
  \act{I} \dummy = \dummy,
  \\
  \begin{aligned}
  \act{I} (\abstr{x} e) &= \abstr{x} (\act{I} e),
  &\quad
  \act{I} (\symM(\vec{t})) &=
  e[\act{I} \vec{t}/\vec{x}] \ \text{if $I(\symM) = \abstr{\vec{x}} e$},
  \\
  \act{I} (\symS(\vec{e}')) &= \symS(\act{I} \vec{e}'),
  &\quad
  \act{I} \symM(\vec{t}) &= \symM(\act{I} \vec{t}) \ \text{if $\symM \not\in |I|$}.
  \end{aligned}
\end{gather*}
The action on abstracted judgements is given by
\begin{align*}
  \act{I} (\isType{A}) &= (\isType{\act{I} A}), \\
  \act{I} (t : A) &= (\act{I} t : \act{I} A), \\
  \act{I} (A \equiv B \bye \dummy) &= (\act{I} A \equiv \act{I} B \bye \dummy), \\
  \act{I} (s \equiv t : A \bye \dummy) &= (\act{I} A \equiv \act{I} B \bye \dummy), \\
  \act{I} (\abstr{x \of A} \; \JJ) &= (\abstr{x \of \act{I} A} \; \act{I} \JJ).
\end{align*}
An abstracted boundary may be instantiated analogously.

Given $I$ of~$\Xi$ over~$\Theta; \Gamma$, and $\Delta = [x_1 \of A_1, \ldots, x_n \of A_n]$ over $\Theta$ such that $|\Gamma| \cap |\Delta| = \emptyset$, we define $\Gamma, \act{I} \Delta$ to be the variable context
\begin{equation*}
  \Gamma, x_1 \of \act{I} A_1, \ldots, x_n \of \act{I} A_n
\end{equation*}
Note that $\act{I} \Delta$ by itself is not a valid variable context.
A judgement $\Xi; \Delta \types \JJ$ may be instantiated to $\Theta; \Gamma, \act{I} \Delta \types \act{I} \JJ$.
A hypothetical boundary can be instantiated analogously.

\subsection{Raw rules}
\label{sec:raw-rules}

An inference rule in type theory is a template that generates a family of closure rules constituting a deductive system. In our setting, a \defemph{raw rule} is a hypothetical judgement of the form $\Theta; \emptyCtx \types \J$, which we display as
\begin{equation*}
  \rawRule{\Theta}{\J}.
\end{equation*}
It is an \defemph{object rule} when $\J$ is an object judgement, and an \defemph{equality rule} when~$\J$ is an equality judgement.
The designation ``raw'' signals that, even though a raw rule is syntactically sensible, it may be quite unreasonable from a type-theoretic point of view.

Given a rule $R = (\rawRule{\symM_1 \of \BB_1, \ldots, \symM_n \of \BB_n}{\plug{\B}{e}})$, along with an instantiation $I = \finmap{\symM_1 \mto e_1, \ldots, \symM_n \mto e_n}$ of its premises over $\Theta; \Gamma$, the \defemph{rule instantiation} $\act{I} R$ is the closure rule $([p_1, \ldots, p_n, q], r)$ where $p_i$ is
\begin{equation*}
  \Theta; \Gamma \types \plug{(\upact{I}{i} \BB_i)}{e_i},
\end{equation*}
$q$ is $\Theta ; \Gamma \types \act{I} \B$,
and $r$ is $\Theta; \Gamma \types \act{I} (\plug{\B}{e})$.
In this way a raw rule generates a family of closure rules, indexed by instantiations.
The premise~$q$ is needed only in various meta-theoretic inductive arguments, as it ensures the well-formedness of the boundary of the conclusion. In practice, we use the ``economic'' variant $([p_1, \ldots, p_n], r)$, which is easily seen to be admissible once \cref{prop:presuppositivity} is established.

\begin{exa}
  We may translate traditional ways of presenting rules to raw rules easily. For example, consider the formation rule for dependent products, which might be written as
  \begin{equation*}
    \infer
    {
      \Gamma \types \isType{A} \\
      \Gamma, x \of A \types \isType{B}
    }{
      \Gamma \types \isType{\Uppi(A, \abstr{x} B)}
    }
  \end{equation*}
  To be quite precise, the above is a \emph{family} of closure rules, indexed by meta-level parameters $\Gamma$, $A$, and $B$ ranging over suitable syntactic entities.
  The template which generates such a family might be written as
  \begin{equation*}
    \infer
    {
      \types \isType{\sym{A}} \\
      x \of \sym{A} \types \isType{\sym{B}(x)}
    }{
      \types \isType{\Uppi(\sym{A}, \abstr{x} \sym{B}(x))}
    }
  \end{equation*}
  Indeed, there is no need to mention~$\Gamma$ because it is always present, and we have replaced the parameters~$A$ and~$B$ with metavariables $\sym{A}$ and $\sym{B}$ (notice the change of fonts) to obtain bona-fide syntactic expressions. Next, observe that the premises amount to specifying an abstracted boundary for each metavariable, which brings us to
  \begin{equation*}
    \infer
    {
      \sym{A} \of (\isType{\Box}) \\
      \sym{B} \of (\abstr{x \of \sym{A}} \; \isType{\Box})
    }{
      \isType{\Uppi(\sym{A}, \abstr{x} \sym{B}(x))}
    }
  \end{equation*}
  By writing everything in a single line we obtain a raw rule
  \begin{equation*}
    \rawRule
    {
      \sym{A} \of (\isType{\Box}),
      \sym{B} \of (\abstr{x \of \sym{A}} \; \isType{\Box})
    }{
      \isType{\Uppi(\sym{A}, \abstr{x} \sym{B}(x))}
    }.
  \end{equation*}
  The original family of closure rules is recovered when the above raw rule is instantiated with $\finmap{\sym{A} \mto A, \sym{B} \mto \abstr{x}B}$ where $A$ and $B$ are type expressions over (a metavariable context and) a variable context $\Gamma$.
\end{exa}

We next define congruence and metavariable rules. These feature in every type theory.

\begin{defi}
  \label{def:congruence-rule}
  The \defemph{congruence rules} associated with a raw object rule~$R$
  \begin{equation*}
    \rawRule{\symM_1 \of \BB_1, \ldots, \symM_n \of \BB_n}{\plug{\B}{e}}
  \end{equation*}
  are closure rules, with
  \begin{equation*}
    I = \finmap{\symM_1 \mto f_1, \ldots, \symM_n \mto f_n}
    \quad\text{and}\quad
    J = \finmap{\symM_1 \mto g_1, \ldots, \symM_n \mto g_n},
  \end{equation*}
  of the form
  \begin{equation*}
    \infer{
      {\begin{aligned}
      &\Theta; \Gamma \types \plug{(\upact{I}{i} \BB_i)}{f_i}  &&\text{for $i = 1, \ldots, n$}\\
      &\Theta; \Gamma \types \plug{(\upact{J}{i} \BB_i)}{g_i}  &&\text{for $i = 1, \ldots, n$}\\
      &\Theta; \Gamma \types \plug{(\upact{I}{i} \BB_i)}{f_i \equiv g_i} &&\text{for object boundary $\BB_i$} \\
      &\Theta; \Gamma \types \act{I} B \equiv \act{J} B        &&\text{if $\B = (\Box : B)$}
    \end{aligned}}
    }{
      \Theta; \Gamma \types \plug{(\act{I} \B)}{\act{I} e \equiv \act{J} e}
    }
  \end{equation*}
\end{defi}

Metavariables have their own formation and congruence rules, akin to specific and congruence rules.

\begin{defi}
  \label{def:metavariable-rule}%
  Given a context $\Theta; \Gamma$ over~$\Sigma$ with
  \begin{equation*}
    \Theta = [\symM_1 \of \BB_1, \ldots, \symM_n \of \BB_n]
    \quad\text{and}\quad
    \BB_k = (\abstr{x_1 \of A_1} \cdots \abstr{x_m \of A_m}\; \B),
  \end{equation*}
  the \defemph{metavariable rules} for $\symM_k$ are the closure rules of the form
  \begin{equation*}
    \infer
    {{\begin{aligned}
     &\Theta; \Gamma \types t_j : A_j[\upto{\vec{t}}{j}/\upto{\vec{x}}{j}]
      &&\text{for $j = 1, \ldots, m$}
     \\
     &\Theta; \Gamma \types \B[\vec{t}/\vec{x}]
    \end{aligned}}
    }{
      \Theta; \Gamma \types \plug{(\B[\vec{t}/\vec{x}])}{\symM_k(\vec{t})}
    }
  \end{equation*}
  where $\vec{x} = (x_1, \ldots, x_m)$ and $\vec{t} = (t_1, \ldots, t_m)$.
  Furthermore, if $\B$ is an object boundary, then the \defemph{metavariable congruence
    rules} for $\symM_k$ are the closure rules of the form
  \begin{equation*}
    \infer
    {
     { \begin{aligned}
          &\Theta; \Gamma \types s_j :
              A_j[\upto{\vec{s}}{j}/\upto{\vec{x}}{j}]
          &\text{for $j = 1, \ldots, m$}
          \\
          &\Theta; \Gamma \types t_j :
              A_j[\upto{\vec{t}}{j}/\upto{\vec{x}}{j}]
          &\text{for $j = 1, \ldots, m$}
          \\
          &\Theta; \Gamma \types s_j \equiv t_j :
              A_j[\upto{\vec{s}}{j}/\upto{\vec{x}}{j}]
          &\text{for $j = 1, \ldots, m$}
          \\
          &\Theta; \Gamma \types C[\vec{s}/\vec{x}] \equiv C[\vec{t}/\vec{x}]
          &\text{if $\B = (\Box : C)$}
        \end{aligned} }
    }{
      \Theta; \Gamma \types
      \plug
      {(\B[\vec{s}/\vec{x}])}
      {\symM_k(\vec{s}) \equiv \symM_k(\vec{t})}
    }
  \end{equation*}
  where
  $\vec{s} = (s_1, \ldots, s_m)$ and
  $\vec{t} = (t_1, \ldots, t_m)$.
\end{defi}

Once again, presuppositions may be elided safely from the above rules to yield the following admissible ``economic''
versions:
\begin{mathpar}
    \infer{
      {\begin{aligned}
      &\Theta; \Gamma \types \plug{(\upact{I}{i} \BB_i)}{f_i}  &&\text{for equation boundary $\BB_i$}\\
      &\Theta; \Gamma \types \plug{(\upact{I}{i} \BB_i)}{f_i \equiv g_i} &&\text{for object boundary $\BB_i$}
    \end{aligned}}
    }{
      \Theta; \Gamma \types \plug{(\act{I} \B)}{\act{I} e \equiv \act{J} e}
    }

  \infer
    {{\begin{aligned}
     &\Theta; \Gamma \types t_j : A_j[\upto{\vec{t}}{j}/\upto{\vec{x}}{j}]
      &&\text{for $j = 1, \ldots, m$}
    \end{aligned}}
    }{
      \Theta; \Gamma \types \plug{(\B[\vec{t}/\vec{x}])}{\symM_k(\vec{t})}
    }

    \infer
    {
     { \begin{aligned}
          &\Theta; \Gamma \types s_j \equiv t_j :
              A_j[\upto{\vec{s}}{j}/\upto{\vec{x}}{j}]
          &\text{for $j = 1, \ldots, m$}
        \end{aligned} }
    }{
      \Theta; \Gamma \types
      \plug
      {(\B[\vec{s}/\vec{x}])}
      {\symM_k(\vec{s}) \equiv \symM_k(\vec{t})}
    }
\end{mathpar}

\subsection{Type theories}
\label{sec:type-theories}

A type theory in its basic form is a collection of rules that generate a deductive system. While this is too permissive a notion from a type-theoretic standpoint, it is nevertheless useful enough to deserve a name.

\begin{defi}
  \label{def:raw-type-theory}
  A \defemph{raw type theory~$T$} over a signature~$\Sigma$ is a family of raw rules over~$\Sigma$, called the \defemph{specific rules} of~$T$.
  The \defemph{associated deductive system} of~$T$ consists of:
  \begin{enumerate}
  \item the \defemph{structural rules} over~$\Sigma$:
    \begin{enumerate}
    \item the \emph{variable}, \emph{metavariable}, and \emph{abstraction} rules (\cref{def:metavariable-rule,fig:struct-rules}),
    \item the \emph{equality} rules, (\cref{fig:equality-rules}),
    \item the \emph{boundary} rules (\cref{fig:well-formed-boundaries});
    \end{enumerate}
  \item the instantiations of the specific rules of~$T$;
  \item for each specific object rule of~$T$, the instantiations of the associated congruence rule (\cref{def:congruence-rule}).
  \end{enumerate}
\end{defi}

\begin{figure}[pt]
  \centering
  \small
  \begin{ruleframe}
  \begin{mathpar}
    \inferenceRule{TT-Var}
    {
      \var a \in |\Gamma|
    }{
      \Theta; \Gamma \types \var{a} : \Gamma(\var{a})
    }

    \inferenceRule{TT-Meta}
    {{\begin{aligned}
     &\mathrlap{\Theta(\symM_k) = \abstr{x_1 \of A_1} \cdots \abstr{x_m \of A_m}\; \B} \\
     &\Theta; \Gamma \types t_j : A_j[\upto{\vec{t}}{j}/\upto{\vec{x}}{j}]
      &&\text{for $j = 1, \ldots, m$}
     \\
     &\Theta; \Gamma \types \B[\vec{t}/\vec{x}]
    \end{aligned}}
    }{
      \Theta; \Gamma \types \plug{(\B[\vec{t}/\vec{x}])}{\symM_k(\vec{t})}
    }

    \inferenceRule{TT-Meta-Congr}
    {
     { \begin{aligned}
          &\Theta(\symM_k) = \abstr{x_1 \of A_1} \cdots \abstr{x_m \of A_m}\; \B\\
          &\Theta; \Gamma \types s_j :
              A_j[\upto{\vec{s}}{j}/\upto{\vec{x}}{j}]
          &\text{for $j = 1, \ldots, m$}
          \\
          &\Theta; \Gamma \types t_j :
              A_j[\upto{\vec{t}}{j}/\upto{\vec{x}}{j}]
          &\text{for $j = 1, \ldots, m$}
          \\
          &\Theta; \Gamma \types s_j \equiv t_j :
              A_j[\upto{\vec{s}}{j}/\upto{\vec{x}}{j}]
          &\text{for $j = 1, \ldots, m$}
          \\
          &\Theta; \Gamma \types C[\vec{s}/\vec{x}] \equiv C[\vec{t}/\vec{x}]
          &\text{if $\B = (\Box : C)$}
        \end{aligned} }
    }{
      \Theta; \Gamma \types
      \plug
      {(\B[\vec{s}/\vec{x}])}
      {\symM_k(\vec{s}) \equiv \symM_k(\vec{t})}
    }

    \inferenceRule{TT-Abstr}
    {
      \Theta; \Gamma \types \isType A \\
      \var{a} \not\in |\Gamma| \\
      \Theta; \Gamma, \var{a} \of A \types \JJ[\var{a}/x]
    }{
      \Theta; \Gamma \types \abstr{x \of A} \; \JJ
    }

  \end{mathpar}
  \end{ruleframe}
  \caption{Variable, metavariable and abstraction closure rules}
  \label{fig:struct-rules}
\end{figure}

\begin{figure}[pht]
  \centering
  \small
  \begin{ruleframe}
  \begin{mathpar}
  \inferenceRule{TT-Ty-Refl}
  { \Theta; \Gamma \types \isType { A } }
  { \Theta; \Gamma \types A \equiv A }

  \inferenceRule{TT-Ty-Sym}
  { \Theta; \Gamma \types A \equiv B }
  { \Theta; \Gamma \types B \equiv A }

  \inferenceRule{TT-Ty-Tran}
  { \Theta; \Gamma \types A \equiv B \\
    \Theta; \Gamma \types B \equiv C }
  { \Theta; \Gamma \types A \equiv C }

  \inferenceRule{TT-Tm-Refl}
  { \Theta; \Gamma \types t : A }
  { \Theta; \Gamma \types t \equiv t : A }

  \inferenceRule{TT-Tm-Sym}
  { \Theta; \Gamma \types s \equiv t : A }
  { \Theta; \Gamma \types t \equiv s : A}

  \inferenceRule{TT-Tm-Tran}
  { \Theta; \Gamma \types s \equiv t : A \\
    \Theta; \Gamma \types t \equiv u : A }
  { \Theta; \Gamma \types s \equiv u : A }

  \inferenceRule{TT-Conv-Tm}
  { \Theta; \Gamma \types t : A \\
    \Theta; \Gamma \types A \equiv B }
  { \Theta; \Gamma \types t : B }

  \inferenceRule{TT-Conv-Eq}
  {
    \Theta; \Gamma \types s \equiv t : A \\
    \Theta; \Gamma \types A \equiv B }
  { \Theta; \Gamma \types s \equiv t : B }
  \end{mathpar}
  \end{ruleframe}
  \caption{Equality closure rules}
  \label{fig:equality-rules}
\end{figure}

\begin{figure}[pht]
  \centering
  \small
  \begin{ruleframe}
  \begin{mathpar}
    \inferenceRule{TT-Bdry-Ty}
    {
    }{
     \Theta; \Gamma \types \isType \Box
    }

    \inferenceRule{TT-Bdry-Tm}
    {
      \Theta; \Gamma \types \isType{A}
    }{
      \Theta; \Gamma \types \Box : A
    }

    \inferenceRule{TT-Bdry-EqTy}
    {
      \Theta; \Gamma \types \isType{A} \\
      \Theta; \Gamma \types \isType{B}
    }{
      \Theta; \Gamma \types A \equiv B \bye \Box
    }

    \inferenceRule{TT-Bdry-EqTm}
    {
      \Theta; \Gamma \types \isType{A} \\
      \Theta; \Gamma \types s : A \\
      \Theta; \Gamma \types t : A
    }{
      \Theta; \Gamma \types s \equiv t : A \bye \Box
    }

    \inferenceRule{TT-Bdry-Abstr}
    {
      \Theta; \Gamma \types \isType A \\
      \var{a} \not\in |\Gamma| \\
      \Theta ; \Gamma, \var{a} \of A \types \BB[\var{a}/x]
    }{
      \Theta; \Gamma \types \abstr{x \of A} \; \BB
    }
  \end{mathpar}
  \end{ruleframe}
  \caption{Well-formed abstracted boundaries}
  \label{fig:well-formed-boundaries}
\end{figure}

\begin{figure}[pht]
  \centering
  \small
  \begin{ruleframe}
  \begin{mathpar}
    \inferenceRule{MCtx-Empty}
    {
    }{
      \types \isExt{\emptyExt}
    }

    \inferenceRule{MCtx-Extend}
    {
      \types \isExt{\Theta} \\
      \Theta ; \emptyCtx \types \BB \\
      \symM \not\in |\Theta|
    }{
      \types \isExt{\finmap{\Theta, \symM \of \BB}}
    }

    \\

    \inferenceRule{VCtx-Empty}
    {
      \types \isExt{\Theta}
    }{
      \Theta; \Gamma \types \isCtx{\emptyCtx}
    }

    \inferenceRule{VCtx-Extend}
    {
      \Theta; \Gamma \types \isCtx{\Delta} \\
      \Theta; \Gamma, \Delta \types \isType{A} \\
      \var{a} \not\in |\Gamma, \Delta|
    }{
      \Theta; \Gamma \types \isCtx{\finmap{\Delta, \var{a} \of A}}
    }
  \end{mathpar}
  \end{ruleframe}
  \caption{Well-formed contexts}
  \label{fig:contexts}
\end{figure}

The rules of a raw type theory do not impose any conditions on the contexts, although they only ever extend variable contexts with well-formed types. When a well-formed context is needed, the auxiliary rules in \cref{fig:contexts} are employed.

\medskip

With the notion of raw type theory in hand, we may define concepts that employ derivability.

\begin{defi}
  An instantiation $I = \finmap{M_1 \mto e_1, \ldots, M_n \mto e_n}$ of a metavariable context $\Xi = [M_1 \of \BB_1, \ldots, M_n \of \BB_n]$ over $\Theta; \Gamma$ is \defemph{derivable} when
  $
  \Theta ; \Gamma \types \plug{(\upact{I}{k} \BB_k)}{e_k}
  $
  for $k = 1, \ldots, n$.
\end{defi}

\begin{defi}
  Instantiations
  \begin{equation*}
    I = \finmap{\symM_1 \mto e_1, \ldots, \symM_n \mto e_n}
    \qquad\text{and}\qquad
    J = \finmap{\symM_1 \mto f_1, \ldots, \symM_n \mto f_n}
  \end{equation*}
  over $\Theta; \Gamma$ are \defemph{judgementally equal} when, for $k = 1, \ldots, n$,
  if $\BB_k$ is an object boundary then
  $
    \Theta; \Gamma \types \plug{(\upact{I}{k}\BB_k)}{e_k \equiv f_k}
  $.
\end{defi}

\begin{defi}
  A raw rule $\rawRule{\Xi}{\J}$ is \defemph{derivable} when it is derivable qua judgement.
  It is \defemph{admissible} when, for every derivable instantiation~$I = \finmap{M_1 \mto e_1, \ldots, M_n \mto e_n}$ of~$\Xi$ over~$\Theta; \Gamma$ we have $\Theta; \Gamma \types \act{I} \J$.
\end{defi}

If $I$ is an instantiation of $\Xi = [\symM_1 \of \BB_1, \ldots, \symM_m \of \BB_m]$ over~$\Theta; \Delta$, and $J$ is an instantiation of $\Theta$ over $\Psi; \Gamma$ such that $|\Gamma| \cap |\Delta| = \emptyset$, their \defemph{composition $J \circ I$} is the instantiation of $\Xi$ over $\Psi; \Gamma, \act{J} \Delta$ defined by
\begin{equation*}
  (J \circ I)(\symM) = \act{J}(I(\symM)).
\end{equation*}
Composition of instantiations is associative. It also preserves derivability, which can be proved easily once \cref{thm:admiss-of-inst} is established.

It will be useful to know that derivability only needs to be checked for instantiations over the empty variable context. For this purpose, define the \defemph{promotion} of
\begin{equation*}
  \Theta; \Gamma \types \JJ
\end{equation*}
to be the judgement
\begin{equation*}
  (\Theta, \Gamma); \emptyCtx \types \JJ,
\end{equation*}
in which the free variables are promoted to metavariables. (We could obfuscate what we just said by being more precise: if $\Gamma = [\var{a}_1 \of A_1, \ldots, \var{a}_n \of A_n]$, the promotion is the judgement
\begin{equation*}
  (\Theta, \sym{a}'_1 \of A'_1, \ldots, \sym{a}'_n \of A'_n); \emptyCtx \types \JJ[\vec{\sym{a}}'/\vec{\var{a}}]
\end{equation*}
in which $\sym{a}'_1, \ldots, \sym{a}'_n$ are fresh and $A'_i = A_i[\upto{\vec{\sym{a'}}}{i}/\upto{\vec{\var{a}}}{i}]$.)
Note that $\types \isExt{(\Theta, \Gamma)}$ is derivable if, and only if, both $\types \isExt{\Theta}$ and $\Theta \types \isCtx{\Gamma}$ are derivable.

\begin{prop}
  \label{prop:promotion}
  A raw type theory derives $\Theta; \Gamma \types \JJ$ if, and only if, it derives the promotion $(\Theta, \Gamma); \emptyCtx \types \JJ$.
\end{prop}

\begin{proof}
  To pass between the original variable context and its promotion, swap applications of \rref{TT-Var} with corresponding applications of \rref{TT-Meta}.
\end{proof}

% \begin{lem}
%   \label{lem:compisitions-of-derivable-inst-is-derivable}
%   Composition of derivable instantiations is a derivable instantiation.
% \end{lem}

% \begin{proof}
%   We need to prove for $i = 1, \ldots, m$
%   \begin{equation}
%     \label{eq:composition-inst-der-1}
%     \Psi ; \Delta \types \plug{(\upact{(J\circ I)}{i} \BB_i)}{\act{J} e_i}
%   \end{equation}
%   By definition we compute $\upto{(J \circ I)}{i} = J \circ (\upto{I}{i})$.
%   Since $I$ is derivable, we know
%   \begin{equation}
%     \label{eq:composition-inst-der-2}
%     \Theta ; \emptyCtx \types \plug{(\upact{I}{i}\BB_i)}{e_i}.
%   \end{equation}
%   Because $J$ is derivable we use~\cref{thm:admiss-of-inst} on \eqref{eq:composition-inst-der-2} to derive~\eqref{eq:composition-inst-der-1}.
% \end{proof}

\medskip

Raw rules do not impose any well-typedness conditions on the premises and the conclusion. We may rectify this by requiring that their boundaries be derivable, as follows.

\begin{defi}
  \label{def:finitary-rule-theory}
  A raw rule $\rawRule{\Theta}{\plug{\B}{e}}$ is a \defemph{finitary rule} with respect to a raw type theory~$T$ when $\types \isExt{\Theta}$ and $\Theta; \emptyCtx \types \B$ are derivable.
  A \defemph{finitary type theory} is a raw type theory $T$ whose rules are finitary with respect to $T$.
\end{defi}

According to the above definition, the justification that a rule is finitary may rely on
the rule itself. If so desired, such circularity may be proscribed by imposition of a
well-found order on the rules, with the requirement that the finitary character of each
rule be established using only the rules preceding it, see~\cite{bauer20:gtt} for further
details.

\begin{exa}
  A finitary type theory is well behaved in many respects, but may still be ``non-standard''. Assuming $\sym{N}$, $\sym{O}$ and $\sym{S}$ are respectively a type constant, a term constant, and a unary term symbol, the rules
  \begin{equation*}
    \rawRule{\emptyExt}{\isType{\sym{N}}},
    \qquad
    \rawRule{\emptyExt}{\sym{O} : \sym{N}},
    \qquad
    \rawRule{\sym{n} \of (\Box : \sym{N})}{\sym{S}(\sym{S}(\sym{n})) : \sym{N}}
  \end{equation*}
  constitute a finitary type theory. However, the third rule is troublesome because it posits a compound term $\sym{S}(\sym{S}(\sym{n}))$.
\end{exa}

 We avoid such anomalies by requiring that object rules only ever introduce generically applied symbols.
For this purpose, define a \defemph{rule-boundary} to be a hypothetical boundary of the form
$
\Theta; \emptyCtx \types \B
$,
notated as
$
  \rawRule{\Theta}{\B}
$.
The elements of~$\Theta$ are the \defemph{premises} and $\B$ is the \defemph{conclusion boundary}.
We say that the rule-boundary is an \defemph{object rule-boundary} when~$\B$ is a type or a term boundary, and an \defemph{equality rule-boundary} when $\B$ is an equality boundary.
Next, given an object rule-boundary
\begin{equation*}
  \rawRule{\symM_1 \of \BB_1, \ldots, \symM_n \of \BB_n}{\B}.
\end{equation*}
its \defemph{associated symbol arity} is $(c, [\arity{\BB_1}, \ldots, \arity{\BB_n}])$, where $c \in \set{\Ty, \Tm}$ is the syntactic class of~$\B$.
Given a fresh symbol~$\sym{S}$, we assign it the associated arity and define the
\defemph{associated symbol rule} to be
\begin{equation*}
  \rawRule
  {\symM_1 \of \BB_1, \ldots, \symM_n \of \BB_n}
  {\B[\symS(\genapp{\symM}_i, \ldots, \genapp{\symM}_n)]},
\end{equation*}
where $\genapp{\symM}_i$ is the \defemph{generic application} of the metavariable~$\symM_i$, defined as, assuming $\arity{\BB_i} = (c_i, n_i)$:
\begin{enumerate}
\item
  $\genapp{M}_i = \abstr{x_1, \ldots, x_{n_i}} \symM_i(x_1, \ldots, x_{n_i})$ when $c_i \in \set{\Ty, \Tm}$,
\item
  $\genapp{M}_i = \abstr{x_1, \ldots, x_{n_i}} \dummy$ when $c_i \in \set{\EqTy, \EqTm}$.
\end{enumerate}
Here then is our final notion of type theory.

\begin{defi}
  \label{def:standard-type-theory}
  A finitary type theory is \defemph{standard} if its specific object rules are symbol rules, and each symbol has precisely one associated rule.
\end{defi}

\subsection{A review of meta-theorems}

We recall from~\cite{bauer:_finit} meta-theorems that establish desirable structural properties of type theories.
In the next section we prove several additional meta-theorems that we rely on subsequently.

First we have meta-theorems about raw type theories that are concerned with syntactic manipulations.

\begin{prop}[Renaming]
  \label{prop:tt-renaming}%
  If a raw type theory derives a judgement or a boundary, then it also derives its renaming.
\end{prop}

\begin{prop}[Weakening]
  \label{prop:tt-weakening}
  For a raw type theory:
  \begin{enumerate}
  \item If $\Theta; \Gamma_1, \Gamma_2 \types \JJ$ then $\Theta; \Gamma_1, \var{a} \of A, \Gamma_2 \types \JJ$.
  \item If $\Theta_1, \Theta_2; \Gamma \types \JJ$ then $\Theta_1, \symM \of \BB, \Theta_2; \Gamma \types \JJ$.
  \end{enumerate}
  An analogous statement holds for boundaries.
\end{prop}

It is understood that in the above statements, and the subsequent ones, we tacitly assume whatever syntactic conditions are needed to ensure that all entities are well-formed. For example, in \cref{prop:tt-weakening} we require $\var{a} \not\in |\Gamma_1, \Gamma_2|$ and that $A$ be a syntactically valid type expression for $\Theta; \Gamma_1$.

\begin{thm}[Admissibility of substitution]
  \label{thm:substitution-admissible}
  In a raw type theory the substitution rules from \cref{fig:substitution-rules} are admissible.
\end{thm}

\begin{figure}[pht]
  \centering
  \small
  \begin{ruleframe}
  \begin{mathpar}
    \inferenceRule{TT-Subst}
    {
      \Theta; \Gamma \types \abstr{x \of A} \; \JJ
      \\
      \Theta; \Gamma \types t : A
    }{
      \Theta; \Gamma \types \JJ[t/x]
    }

    \inferenceRule{TT-Bdry-Subst}
    {
      \Theta; \Gamma \types \abstr{x \of A} \; \BB
      \\
      \Theta; \Gamma \types t : A
    }{
      \Theta; \Gamma \types \BB[t/x]
    }

    \inferenceRule{TT-Subst-EqTy}
    {
      \Theta; \Gamma \types \abstr{x \of A} \abstr{\vec{y} \of \vec{B}} \; \isType{C} \\
      \Theta; \Gamma \types s : A \\
      \Theta; \Gamma \types t : A \\
      \Theta; \Gamma \types s \equiv t : A
    }{
      \Theta; \Gamma \types \abstr{\vec{y} \of \vec{B}[s/x]} \; C[s/x] \equiv C[t/x]
    }

    \inferenceRule{TT-Subst-EqTm}
    {
      \Theta; \Gamma \types \abstr{x \of A} \abstr{\vec{y} \of \vec{B}} \; u : C \\
      \Theta; \Gamma \types s : A \\
      \Theta; \Gamma \types t : A \\
      \Theta; \Gamma \types s \equiv t : A
    }{
      \Theta; \Gamma \types \abstr{\vec{y} \of \vec{B}[s/x]} \; u[s/x] \equiv u[t/x] : C[s/x]
    }

    \inferenceRule{TT-Conv-Abstr}
    {
      \Theta; \Gamma \types \abstr{x \of A}\; \JJ
      \\ \Theta; \Gamma \types \isType B \\ \Theta; \Gamma \types A \equiv B }
    { \Theta; \Gamma \types \abstr{x \of B}\; \JJ }
  \end{mathpar}
  \end{ruleframe}
  \caption{Admissible substitution rules}
  \label{fig:substitution-rules}
\end{figure}

Next we have admissibility of instantiations.

\begin{thm}[Admissibility of instantiations]
  \label{thm:admiss-of-inst}%
  In a raw type theory, let $I$ be a derivable instantiation of $\Xi$ over~$\Theta; \Delta$.
  If $\Xi; \Gamma \types \JJ$ is derivable and $|\Delta| \cap |\Gamma| = \emptyset$ then
  $\Theta; \Delta, \act{I} \Gamma \types \act{I} \JJ$ is derivable, and similarly for boundaries.
\end{thm}

\begin{thm}
  \label{thm:eq-inst-admit}%
  \label{thm:instEq}
  In a raw type theory, let $I$ and $J$ be judgementally equal derivable instantiations
  of~$\Xi$ over~$\Theta; \Gamma$.
  Suppose $\Xi \types \isCtx{\Delta}$ and $|\Gamma| \cap |\Delta| = \emptyset$.
  If $\Xi; \Delta \types \plug{\BB}{e}$ is a derivable object judgement then
  $\Theta; \Gamma, \act{I} \Delta \types \plug{(\act{I} \BB)}{\act{I} e \equiv \act{J} e}$ is derivable.
\end{thm}

To obtain meta-theorems with genuine type-theoretic contents, we need to restrict to finitary type theories.

\begin{thm}[Presuppositivity]%
  \label{prop:presuppositivity}
  If a finitary type theory derives $\Theta \types \isCtx{\Gamma}$, and $\Theta; \Gamma \types \plug{\BB}{e}$ then it derives $\Theta; \Gamma \types \BB$.
\end{thm}

The next two theorems apply to standard type theories. The first one provides an inversion principle, and the second one guarantees that a term has at most one type, up to judgemental equality. Both rely on a candidate type that may be read off directly from the term.

\begin{defi}
  Let $T$ be a standard type theory.
  The \defemph{natural type} $\natty{\Theta; \Gamma}{t}$ of a term expression~$t$ with respect to a context~$\Theta; \Gamma$ is defined by:
  \begin{align*}
    \natty{\Theta; \Gamma}{\var{a}}
    &= \Gamma(a), \\
    \natty{\Theta; \Gamma}{\symM(t_1, \ldots, t_m)}
    &=
      \begin{aligned}[t]
      &A[t_1/x_1, \ldots, t_m/x_m] \\
      &\qquad \text{where $\Theta(\symM) = (\abstr{x_1 \of A_1} \cdots \abstr{x_m \of A_m} \; \Box : A)$}
      \end{aligned}
    \\
    \natty{\Theta; \Gamma}{\symS(e_1, \ldots, e_n)}
    &=
      \begin{aligned}[t]
      &\finmap{\symM_1 \mto e_1, \ldots, \symM_n \mto e_n}_* B\\
      &\qquad
         \begin{aligned}[t]
           &\text{where the symbol rule for $\symS$ is}\\
           &\text{$\rawRule{\symM_1 \of \BB_1, \ldots, \symM_n \of \BB_n}
                            {\symS(\genapp{\symM}_1, \ldots, \genapp{\symM}_n): B}$.}
      \end{aligned}
    \end{aligned}
  \end{align*}
\end{defi}

The following theorem is an inversion principle that recovers the ``stump'' of a derivation of a derivable object judgement.

\begin{thm}[Inversion theorem]
  \label{thm:inversion}
  If a standard finitary type theory derives a term judgement then it does so by a derivation which concludes with precisely one of the following rules:
  \begin{enumerate}
  \item the variable rule \rulename{TT-Var},
  \item the metavariable rule \rulename{TT-Meta},
  \item an instantiation of a symbol rule,
  \item the abstraction rule \rulename{TT-Abstr},
  \item the term conversion rule \rulename{TT-Conv-Tm} of the form
    \begin{equation*}
      \infer
      {\Theta; \Gamma \types t : \natty{\Theta;\Gamma}{t} \\
       \Theta; \Gamma \types \natty{\Theta; \Gamma}{t} \equiv A}
      {\Theta; \Gamma \types t : A}
    \end{equation*}
    where $\natty{\Theta;\Gamma}{t} \neq A$.
  \end{enumerate}
\end{thm}

Finally, in a standard type theory a term has at most one type, up to judgemental equality.

\begin{thm}[Uniqueness of typing]%
  \label{thm:uniqueness-of-typing}%
  For a standard finitary type theory:
  \begin{enumerate}
  \item If $\Theta; \Gamma \types t : A$ and $\Theta; \Gamma \types t : B$ then $\Theta; \Gamma \types A \equiv B$.
  \item If $\Theta; \emptyCtx \types \isCtx{\Gamma}$ and $\Theta; \Gamma \types s \equiv t : A$ and $\Theta; \Gamma \types s \equiv t : B$ then $\Theta; \Gamma \types A \equiv B$.
  \end{enumerate}
\end{thm}

\subsection{More meta-theorems}
\label{sec:more-meta-theor}

We state and prove several further meta-theorems.

\begin{prop}
  \label{prop:natty-inst}
  Let $T$ be a standard type theory and $I$ an instantiation of~$\Xi$ over $\Theta; \Gamma$. For a term expression $\symS(\vec{e})$ it holds that $\act{I}(\natty{\Xi ; \Delta}{\symS(\vec{e})}) = \natty{\Theta ; \Gamma, \act{I} \Delta}{\act{I} \symS(\vec{e})}$.
\end{prop}

\begin{proof}
  Let $\rawRule{\symM_1 \of \BB_1, \ldots, \symM_n \of \BB_n}{\symS(\genapp{\symM}_1, \ldots, \genapp{\symM}_n): B}$ be the symbol rule for $\symS$. By unfolding the definition of the natural type we have
  \begin{align*}
    \act{I}(\natty{\Xi ; \Delta}{\symS(\vec{e})})
    & = \act{I}(\act{\finmap{\symM_1 \mto e_1, \ldots, \symM_n \mto e_n}} B)
    = \act{\finmap{\symM_1 \mto \act{I}e_1, \ldots, \symM_n \mto \act{I} e_n}} B\\
    & = \natty{\Theta; \Gamma, \act{I} \Delta}{\symS(\act{I} \vec{e})}
    = \natty{\Theta; \Gamma, \act{I} \Delta}{\act{I} (\symS(\vec{e}))} \qedhere
  \end{align*}
\end{proof}
Note that $I$ acts purely syntactically and needs not be derivable for the equation to hold. It is also worth pointing out that the equation does not hold for metavariable term expressions.
We now explicate two common usages of \cref{thm:inversion}.

\begin{cor}
  \label{cor:args-derivable}
  In a standard type theory, suppose the rule for~$\symS$ is
  \begin{equation*}
    \rawRule
    {\symM_1 \of \BB_1, \ldots, \symM_n \of \BB_n}
    {\plug{\B'}{\symS(\genapp{\symM}_1, \ldots, \genapp{\symM}_n)}}.
  \end{equation*}
  If the theory derives $\Theta; \Gamma \types \plug{\B}{\symS(\vec{e})}$
  then it also derives $\Theta; \Gamma \plug{(\upact{I}{i} \BB_i)}{e_i}$ for all $i = 1, \ldots, n$,
  where $I = \finmap{\symM_1 \mto e_1, \ldots, \symM_n \mto e_n}$.
\end{cor}

\begin{proof}
  By \cref{thm:inversion} the judgement is derived by an application of the symbol rule for~$\symS$, possibly followed by a conversion, whose premises are precisely the judgements of interest.
\end{proof}

\begin{cor}
  \label{cor:natty-derivable}%
  If a standard type theory derives $\Theta; \Gamma \types t : A$ then it also derives $\Theta; \Gamma \types t : \natty{\Theta; \Gamma}{t}$.
\end{cor}

\begin{proof}
  By \cref{thm:inversion}, either $A = \natty{\Theta, \Gamma}{t}$ and there is nothing to prove, or
  the derivations ends with
  \begin{equation*}
    \infer
    {\Theta; \Gamma \types t : \natty{\Theta;\Gamma}{t} \\
      \Theta; \Gamma \types \natty{\Theta; \Gamma}{t} \equiv A}
    {\Theta; \Gamma \types t : A}
  \end{equation*}
  which contains the desired equality as a subderivation.
\end{proof}

We next prove a statement about instantiations that needs a couple of preparatory lemmas.

\begin{lem}
  \label{thm:presup-nested}%
  If a finitary type theory derives $\Theta \types \isCtx{\Gamma}$ and $\Theta; \Gamma \types \plug{\BB}{e \equiv e'}$ then it derives $\Theta; \Gamma \types \plug{\BB}{e'}$.
\end{lem}

\begin{proof}
  We proceed by induction on the number of abstractions in the object boundary~$\BB$.

  \inCaseText{$\BB = (\isType{\Box})$, $e = A$ and $e' = B$}
  \Cref{prop:presuppositivity} applied to the assumption $\Theta; \Gamma \types A \equiv B \bye \dummy$
  gives $\Theta; \Gamma \types A \equiv B \bye \Box$, from which $\Theta; \Gamma \types \isType{B}$ follows by inversion.

  \inCaseText{$\BB = (\Box : A)$}
  Similar to the previous case.

  \inCaseText{$\BB = (\abstr{x \of A} \; \BB')$}
  Inversion on the assumption $\Theta; \Gamma \types \abstr{x \of A} \; \plug{\BB'}{e \equiv e'}$ gives
  \begin{equation*}
    \Theta; \Gamma \types \isType{A}
    \qquad\text{and}\qquad
    \Theta; \Gamma, \var{a} \of A \types \plug{(\BB'[\var{a}/x])}{e[\var{a}/x] \equiv e'[\var{a}/x]}.
  \end{equation*}
  By induction hypothesis, the second judgement entails
  \begin{equation*}
    \Theta; \Gamma, \var{a} \of A \types \plug{(\BB'[\var{a}/x])}{e'[\var{a}/x]},
  \end{equation*}
  which we may abstract to the desired form.
\end{proof}

\begin{lem}
  \label{thm:bdry-fill-eq-stuff}%
  In a finitary type theory, consider judgementally equal derivable instantiations~$I$ and~$J$ of~$\Xi$ over $\Theta; \Gamma$, and suppose
  $\Xi \types \isCtx{\Delta}$
  and
  $\Xi; \Delta \types \BB$
  such that $|\Delta| \cap |\Gamma| = \emptyset$.
  If $\Theta; \Gamma, \act{I} \Delta \types \plug{(\act{I} \BB)}{e}$ is derivable then so is $\Theta; \Gamma, \act{I} \Delta \types \plug{(\act{J} \BB)}{e}$.
\end{lem}

\begin{proof}
  We proceed by structural induction on the derivation of $\Xi; \Delta \types \BB$.

  \inCase{TT-Bdry-Ty}
  Trivial because $\act{I} \BB = (\isType{\Box}) = \act{J} \BB$.

  \inCase{TT-Bdry-Tm}
  If the derivation ends with
  \begin{equation*}
    \infer
    {
      \Xi; \Delta \types \isType{A}
    }{
      \Xi; \Delta \types \Box : A
    }
  \end{equation*}
  then $\Theta; \Gamma, \act{I} \Delta \types \act{I} A \equiv \act{J} A$ by \cref{thm:instEq} applied to the premise,
  hence we may convert $\Theta; \Gamma, \act{I} \Delta \types e : \act{I} A$ to
  $\Theta; \Gamma, \act{I} \Delta \types e : \act{J} A$.

  \inCase{TT-Bdry-EqTy}
  If the derivation ends with
  \begin{equation*}
    \infer
    {
      \Xi; \Delta \types \isType{A} \\
      \Xi; \Delta \types \isType{B}
    }{
      \Xi; \Delta \types A \equiv B \bye \Box
    }
  \end{equation*}
  then \cref{thm:instEq} applied to the premises gives us
  \begin{equation*}
    \Theta; \Gamma, \act{I} \Delta \types \act{I} A \equiv \act{J} A
    \qquad\text{and}\qquad
    \Theta; \Gamma, \act{I} \Delta \types \act{I} B \equiv \act{J} B.
  \end{equation*}
  Together with the assumption $\Theta; \Gamma, \act{I} \Delta \types \act{I} A \equiv \act{I} B$, these suffice to derive $\Theta; \Gamma, \act{I} \Delta \types \act{J} A \equiv \act{J} B$.
  % We use symmetry and transitivity on appropriate equations.

  \inCase{TT-Bdry-EqTm}
  similar to \rulename{TT-Bdry-EqTy}.

  \inCase{TT-Bdry-Abstr}
  Suppose $e = \abstr{x} e'$ and the derivation ends with
  \begin{equation*}
    \infer
    {
      \Xi; \Delta \types \isType A \\
      \var{a} \not\in |\Delta| \\
      \Xi; \Delta, \var{a} \of A \types \BB'[\var{a}/x]
    }{
      \Xi; \Delta \types \abstr{x \of A} \; \BB'
    }
  \end{equation*}
  where we may assume $\var{a} \not\in |\Gamma|$ without loss of generality.
  \Cref{thm:instEq} applied to the first premise derives
  \begin{equation}
    \label{eq:bdry-fill-eq-stuff-1}
    \Theta; \Gamma, \act{I} \Delta \types \act{I} A \equiv \act{J} A.
  \end{equation}
  By inverting the assumption $\Theta; \Gamma, \act{I} \Delta \types \abstr{x \of \act{I} A} \, \plug{(\act{I}\BB')}{e}$, and possibly renaming a free variable to~$\var{a}$, we obtain
  \begin{equation*}
    \Theta; \Gamma, \act{I} \Delta \types \isType{\act{I} A}
    \quad\text{and}\quad
    \Theta; \Gamma, \act{I} \Delta, \var{a} \of \act{I} A \types (\plug{(\act{I}\BB')}{e})[\var{a}/x].
  \end{equation*}
  Then the induction hypothesis for the second premise yields
  \begin{equation*}
    \Theta; \Gamma, \act{I} \Delta, \var{a} \of \act{I} A \types (\plug{(\act{J}\BB')}{e})[\var{a}/x],
  \end{equation*}
  which we may abstract to
  $
    \Theta; \Gamma, \act{I} \Delta \types \abstr{x \of \act{I} A} \, \plug{(\act{J}\BB')}{e}
  $
  and apply \rref{TT-Conv-Abstr} to convert it to the desired judgement
  $
    \Theta; \Gamma, \act{I} \Delta \types \abstr{x \of \act{J} A} \, \plug{(\act{J}\BB')}{e}
  $.
  The premise $\Theta; \Gamma, \act{I} \Delta \types \isType{\act{J} A}$ is derived
  by \cref{prop:presuppositivity} from~\eqref{eq:bdry-fill-eq-stuff-1}.
\end{proof}

\begin{prop}
  \label{lem:instantiations-equal-derivable}%
  In a finitary type theory, consider instantiations $I$ and $J$ of~$\Xi$ over $\Theta; \Gamma$,
  such that $\types \isExt{\Xi}$ and $\Theta \types \isCtx{\Gamma}$. If $I$ is derivable and~$I$ and~$J$ are judgementally equal then~$J$ is derivable.
\end{prop}

\begin{proof}
  We prove the claim by induction on the length of~$\Xi$. The base case is trivial.
  For the induction step we assume the statement, and show that is still holds when we extend $\Xi$, $I$ and $J$ by one more entry. Specifically, assume that
  \begin{equation}
    \label{eq:inst-eq-der-1}%
    \Xi; \emptyCtx \types \BB,
    \qquad\text{and}\qquad
    \Theta; \Gamma \types \plug{(\act{I} \BB)}{e},
  \end{equation}
  and if $\BB$ is an object boundary also that
  \begin{equation}
    \label{eq:inst-eq-der-2}%
    \Theta; \Gamma \types \plug{(\act{I} \BB)}{e \equiv e'}.
  \end{equation}
  Then we must demonstrate $\Theta; \Gamma \types \plug{(\act{J} \BB)}{e'}$.

  If~$\BB$ is an equality boundary then
  applying \cref{thm:bdry-fill-eq-stuff} to~\eqref{eq:inst-eq-der-1} gives
  $
    \Theta; \Gamma \types \plug{(\act{J} \BB)}{e}
  $,
  and we are done because $e$ and $e'$ are the same.

  If~$\BB$ is an object boundary then applying
  \cref{thm:bdry-fill-eq-stuff} to $\Xi; \emptyCtx \types \BB$ and~\eqref{eq:inst-eq-der-2} gives
  $
    \Theta; \Gamma \types \plug{(\act{J} \BB)}{e \equiv e'}
  $.
  The derivability of $\Theta; \Gamma \types \plug{(\act{J} \BB)}{e'}$ now follows from~\cref{thm:presup-nested}.
\end{proof}

%%% Local Variables:
%%% mode: latex
%%% TeX-master: "equality-checking"
%%% End:
 % background material on type theory
\section{The equality checking algorithm}
\label{sec:type-direct-equal}

The equality checking algorithm applies inference rules by pattern matching them against (parts of) the target equation. We therefore begin by studying the type-theoretic and syntactic properties of rules by which the soundness of pattern matching can be ensured.

% chktex-file 46
% chktex-file 45
% chktex-file 1
% chktex-file 26
% chktex-file 24
\subsection{Patterns and object-invertible rules}
\label{sec:invertible-rules}

In order to derive $\Theta; \Gamma \types \J'$ with the rule $\rawRule{\Xi}{\J}$, we must find an instantiation~$I$ of $\Xi$ over~$\Theta; \Gamma$ such that $\act{I} \J = \J'$. We shall be primarily interested in rules where such~$I$ is unique, when it exists.

\begin{defi}
  \label{def:deterministic-rule}
  A raw rule $\rawRule{\Xi}{\J}$ is \defemph{deterministic} when for every judgement $\Theta; \Gamma \types \J'$ there exists at most one instantiation~$I$ of~$\Xi$ over~$\Theta; \Gamma$ such that $\act{I} \J = \J'$, called a \defemph{matching instantiation}.
\end{defi}

We refrain from trying to characterize the deterministic rules, and instead observe that, given a deterministic rule
\begin{equation*}
  R = (\rawRule {\symM_1 \of \BB_1, \ldots, \symM_n \of \BB_n} {\J})
\end{equation*}
and a judgement $\Theta; \Gamma \types \J'$ we may algorithmically compute~$I$ such that $\act{I} \J = \J'$, or decide that it does not exist.
First of all, every object metavariable of~$R$ must appear in~$\J$, or else $R$ would match in multiple ways the judgement $\Theta, \Theta'; \emptyCtx \types \J$, where $\Theta'$ is a copy of $\Theta$ in which each~$\symM_i$ is replaced with~$\symM'_i$.
Therefore, for any instantiation
\begin{equation*}
  I = \finmap{\symM_1 \mto e_1, \ldots, \symM_n \mto e_n}
\end{equation*}
where $\arity{e_i} = \arity{\BB_i} = (c_i, n_i)$ and $e_i = \abstr{x_1, \ldots, x_{n_i}} e'_i$, the size of~$\act{I} \J$ equals or exceeds the size of each $e'_i$.
We may therefore look for an instantiation that matches $\Theta; \Gamma \types \J$ by
exhaustively searching through all $e'_i$'s over $\Theta; \Gamma$ whose sizes are bounded by the size of~$\J$, of which there are only finitely many.
Of course, we are not suggesting that anyone should use such an exhaustive search in practice.
Instead, we provide a simple syntactic criterion that makes a rule deterministic and easy to match against.

\begin{defi}
  \label{def:pattern}
  \defemph{Patterns} are expressions in which every metavariable occurs at most once either in an application without arguments~$\symM()$, or in an argument of the form $\abstr{\vec{x}} \symM(\vec{x})$, where $\vec{x}$ are the only bound variables in scope. They are described by the grammar in \cref{fig:syntax-patterns}.
\begin{figure}[h]
  \centering
  \small
  \begin{ruleframe}
  \begin{align*}
  \text{Type pattern}\ P
  \bnfis& \symM()
  \bnfor \symS(q_1, \ldots, q_n)           && \text{if $\mv{q_i} \cap \mv{q_j} = \emptyset$ for $i \neq j$}
  \\
  \text{Term pattern}\ p
  \bnfis& \symS(q_1, \ldots, q_n)           && \text{if $\mv{q_i} \cap \mv{q_j} = \emptyset$ for $i \neq j$}
  \\
  \text{Argument pattern}\ q
  \bnfis& \abstr{\vec{x}} \symM(\vec{x})
  \bnfor  P
  \bnfor  p
  \end{align*}
  \end{ruleframe}
  \caption{The syntax of patterns.}
  \label{fig:syntax-patterns}
\end{figure}
\end{defi}
Note that $\symM()$ can only appear as a type pattern, but not as a term pattern. The reason for this lies in the definitions of computation rules (\cref{def:type-computation-rule,def:term-computation-rule}) which we shall see later on.

As defined, the patterns are \emph{linear} in the sense that a metavariable cannot appear several times, and \emph{first-order} because patterns may not appear under abstractions. Non-linearity is not an essential limitation, as we shall see shortly. The restriction to first-order patterns arises because in general a standard type theory may not satisfy the \emph{strengthening} principle which states that if $\types \abstr{x \of A} \JJ$ is derivable and $x \notin \bv{\JJ}$ then $\types \JJ$ is derivable. The principle allows a higher-order pattern to safely extract an expression from within an abstraction, so long as no bound variables escape their scopes.

\begin{exa}
  The head of the conclusion of a symbol rule
  \begin{equation*}
    \rawRule
    {\symM_1 \of \BB_1, \ldots, \symM_n \of \BB_n}
    {\plug{\B}{\symS(\genapp{M}_1, \ldots, \genapp{M}_n)}}
  \end{equation*}
  is a pattern because $\genapp{\symM}_i$ has required form $\abstr{x} \symM_i(\vec{x})$.
\end{exa}

\begin{exa}
  \label{exa:linearize}
  Consider the $\beta$-rule for the first projection from a binary product:
  \begin{equation*}
    \infer{
      \types \isType{\sym{A}}
      \\
      \types \isType{\sym{B}}
      \\
      \types \sym{s} : \sym{A}
      \\
      \types \sym{t} : \sym{B}
    }{
      \types
      \sym{fst}(\sym{A}, \sym{B}, \sym{pair} (\sym{A}, \sym{B}, \sym{s}, \sym{t}))
      \equiv
      \sym{s}
      : \sym{A}
    }
  \end{equation*}
  The left-hand side of the conclusion is not a pattern because the metavariables~$\sym{A}$ and~$\sym{B}$ occur twice each. We may linearize the pattern at the cost of equational premises:
  \begin{equation}
    \label{eq:fst-computation}
    \infer{
      \types \isType{\sym{A}_1}
      \\
      \types \isType{\sym{A}_2}
      \\
      \types \isType{\sym{B}_1}
      \\
      \types \isType{\sym{B}_2}
      \\\\
      \types \sym{s} : \sym{A}_2
      \\
      \types \sym{t} : \sym{B}_2
      \\
      \types \sym{A}_1 \equiv \sym{A}_2
      \\
      \types \sym{B}_1 \equiv \sym{B}_2
    }{
      \types
      \sym{fst}(\sym{A}_1, \sym{B}_1, \sym{pair} (\sym{A}_2, \sym{B}_2, \sym{s}, \sym{t}))
      \equiv
      \sym{s}
      : \sym{A}_1
    }
  \end{equation}
  The new rule is inter-derivable with the original one.
  It is clear that the technique works generally and that it can be automated.
\end{exa}

\begin{exa}
  Consider the rule stating that the identity function is the neutral element for composition:
  \begin{equation*}
    \infer{
      \types \isType{\sym{A}}
      \\
      \types \isType{\sym{B}}
      \\
      \types \sym{f} : \sym{A} \to \sym{B}
    }{
      \types \sym{compose}(\sym{A}, \sym{B}, \sym{f}, \uplambda(\sym{A}, \sym{A}, \abstr{x} x))
             \equiv
             \sym{f} : \sym{A} \to \sym{B}
    }
  \end{equation*}
  The left-hand side of the conclusion is not a pattern because $\uplambda(\sym{A}, \sym{A}, \abstr{x} x)$ is not a pattern. Once again we can remedy the situation by introducing an additional equational premise:
  \begin{equation*}
    \infer{
      \types \isType{\sym{A}}
      \\
      \types \isType{\sym{B}}
      \\
      \types \sym{f} : \sym{A} \to \sym{B}
      \\\\
      \types \sym{i} : \sym{A} \to \sym{A}
      \\
      \types \sym{i} \equiv \lambda(\sym{A}, \sym{A}, \abstr{x} x) : \sym{A} \to \sym{A}
    }{
      \types \sym{compose}(\sym{A}, \sym{B}, \sym{f}, i)
             \equiv
             \sym{f} : \sym{A} \to \sym{B}
    }
  \end{equation*}
\end{exa}

\begin{prop}
  \label{prop:rule-deterministic}
  If $\rawRule{\Xi}{\plug{\B}{p}}$ is a rule such that~$p$ is a pattern and $\mv{p} = \objmv{\Xi}$ then the rule is deterministic.
\end{prop}

\begin{proof}
  Consider a judgement $\Theta; \Gamma \types \plug{\B'}{e}$, and instantiations $J$ and $K$ of $\Xi$ over $\Theta; \Gamma$ such that $\act{J} p = \act{K} p = e$.
  Then $J$ and $K$ agree on object metavariables because they all appear in~$p$, and on equational metavariables because they must.
\end{proof}

We shall use patterns to find matching instantiations, when they exist. For this purpose we define the following notation.

\begin{defi}
  \label{def:matches}%
  Given $\Xi$, a pattern~$p$ over~$\Xi$ such that $\mv{p} = \objmv{\Xi}$, and an expression $e$ over~$\Theta; \Gamma$, we write
  \begin{equation*}
    \matches{\Xi}{p}{t}{I}
    \qquad\text{and}\qquad
    \notmatches{\Xi}{p}{t}
  \end{equation*}
  respectively when $I$ is an instantiation of~$\Xi$ over $\Theta; \Gamma$ such that $\act{I} p = t$, and when there is no such instantiation.
\end{defi}

The reader should convince themselves that there is an obvious algorithm that computes from $\Xi$, $p$ and $t$ the unique~$I$ such that $\matches{\Xi}{p}{t}{I}$, or decides that it does not exist.

Rules are used not only to derive judgements, but also to \emph{invert} derivable judgements to their premises, for the purpose of analyzing them. For example, when a term is normalized, we decide what steps to take by observing its structure, which amounts to applying an inversion principle, such as \cref{thm:inversion}. In general, we may invert a derivable judgement $\Theta; \Gamma \types \J'$ using a rule
\begin{equation*}
  R = (\rawRule{\symM_1 \of \BB_1, \ldots, \symM_n \of \BB_n}{\J})
\end{equation*}
by finding a derivable instantiation~$I = \finmap{\symM_1 \mto e_1, \ldots, \symM_n \mto e_n}$ of its premises over~$\Theta; \Gamma$ such that $\act{I} \J = \J'$. If~$I$ is found, the judgement can be derived using the instantiation $\act{I}R$,
\begin{equation*}
  \infer{
    \Theta; \Gamma \types \plug{(\upact{I}{k} \BB_k)}{e_k}
    \ \text{for $k = 1, \ldots, n$}
  }{
    \Theta; \Gamma \types \act{I} \J
  }
\end{equation*}
Under favorable conditions, it may happen that some of the above premises are known to be derivable ahead of time, so there is no need to rederive them. We are particularly interested in the case where all the object premises are of this kind.

\begin{defi}
  \label{def:equational-residue}
  Let $\Xi = [\symM_1 \of \BB_1, \ldots, \symM_n \of \BB_n]$ be a metavariable context whose equational metavariables are $\symM_{i_1}, \ldots, \symM_{i_m}$. Given an instantiation $I$ of~$\Xi$ over~$\Theta; \emptyCtx$ such that $|\Xi| \cap |\Theta| = \emptyset$, the \defemph{equational residue} $\residue{\Xi}{I}$ is the metavariable context
  \begin{equation*}
    \residue{\Xi}{I} =
      [\Theta,
       \symM_{i_1} \of \upact{I}{i_1} \BB_{i_1}, \ldots, \symM_{i_m} \of \upact{I}{i_m} \BB_{i_m}].
  \end{equation*}
  The \defemph{residual instantiation~$\residueInst{I}$} of $\Xi$ over $\residue{\Xi}{I}$ and $\emptyCtx$ is defined by
  \begin{equation*}
     \residueInst{I}(\symM_i) =
     \begin{cases}
       I(\symM_i)       & \text{if $\symM_i \in \objmv{\Xi}$,}\\
       \genapp{\symM_i} & \text{otherwise.}
     \end{cases}
  \end{equation*}
\end{defi}

\begin{defi}
  \label{def:invertible-rule}
  In a raw type theory, a derivable raw rule
  $R = (\rawRule{\Xi}{\J})$
  is \defemph{object-invertible} when the following holds:
  whenever $I$ instantiates $\Xi$ over~$\Theta; \emptyCtx$,
  with $\types \isExt{\Theta}$ and $|\Xi| \cap |\Theta| = \emptyset$,
  if $\Theta; \emptyCtx \types \act{I}{\J}$ is derivable then so is the residual instantiation~$\residueInst{I}$.
\end{defi}

Let us explain how object-invertible rules shall be used. Suppose $\rawRule{\Xi}{s : A}$ is object-invertible, $\rawRule{\Xi}{s \equiv t : A}$ is derivable, $I$ instantiates $\Xi$ over $\Theta; \Gamma$, and $\Theta; \Gamma \types \act{I} s : \act{I} A$ is given. We would like to derive $\Theta; \Gamma \types \act{I} s \equiv \act{I} t : \act{I} A$ so that we may rewrite $\act{I} s$ to $\act{I} t$. Thus we must verify that~$I$ is derivable.
By object-invertibility its object premises are derivable, so we only need to check its equational ones.
The following proposition ensures that such a procedure is valid.

\begin{prop}
  \label{prop:object-invertible-promote}
  Consider an object-invertible rule $\rawRule{\Xi}{\J}$ and an instantiation $I$ over $\Theta; \Gamma$,
  such that $\Theta; \Gamma \types \act{I} \J$ is derivable. Then $I$ is derivable if, for every equational boundary
  $\BB = \abstr{\vec{x} \of \vec{A}} \; \B$ in $\Xi$, the judgement
  $
    \Theta; \Gamma \types \plug{(\act{I}\BB)}{\abstr{\vec{x}}\dummy}
  $
  is derivable.
\end{prop}

\begin{proof}
  Let $J$ be the promotion of~$I$ to $(\Theta, \Gamma)$ and the empty context.
  Because the rule is object-invertible, $\residueInst{J}$ is derivable.
  Next, we promote each judgement from the statement to
  \begin{equation}
    \label{eq:object-invertible-promote-premises}
    (\Theta, \Gamma) ; \emptyCtx \types \plug{(\act{J}\BB)}{\abstr{\vec{x}} \dummy}.
  \end{equation}
  and observe that $J = K \circ \residueInst{J}$, where $K$ is the instantiation of
  $\residue{\Xi}{J}$ over $(\Theta, \Gamma)$ and $\emptyCtx$ defined by
  \begin{equation*}
    K(\sym{M}) =
    \begin{cases}
      \genapp{\symM}         & \text{if $\symM \in |(\Theta, \Gamma)|$,} \\
      \abstr{\vec{x}} \dummy & \text{otherwise.}
    \end{cases}
  \end{equation*}
  Because $\residueInst{J}$ is derivable, and $K$ is derivable thanks to derivability of judgements~\eqref{eq:object-invertible-promote-premises}, it follows that~$J$ is derivable.
  Therefore, $I$ is derivable too.
\end{proof}

\begin{exa}
  Let us demonstrate how equational residues are going to be used in rewriting.
  Suppose we have derived
  \begin{equation}
    \label{eq:example-residue}
    \Theta; \emptyCtx \types
    \sym{fst}(U_1, V_1, \sym{pair} (U_2, V_2, u, v))
    : U_1
  \end{equation}
  and would like to apply the $\beta$-rule~\eqref{eq:fst-computation} to it, i.e., we
  would like to establish
  \begin{equation}
    \label{eq:example-residue-eq}
    \Theta; \emptyCtx \types
    \sym{fst}(U_1, V_1, \sym{pair} (U_2, V_2, u, v)) \equiv u
    : U_1
  \end{equation}
  First, using \Cref{prop:presuppositivity}, we extract from~\eqref{eq:fst-computation}
  the derivability of its left-hand side
  \begin{equation}
    \label{eq:example-residue-lhs}
    \infer{
      \types \isType{\sym{A}_1}
      \\
      \types \isType{\sym{A}_2}
      \\
      \types \isType{\sym{B}_1}
      \\
      \types \isType{\sym{B}_2}
      \\\\
      \types \sym{s} : \sym{A}_2
      \\
      \types \sym{t} : \sym{B}_2
      \\
      \types \sym{A}_1 \equiv \sym{A}_2 \bye \zeta
      \\
      \types \sym{B}_1 \equiv \sym{B}_2 \bye \xi
    }{
      \types
      \sym{fst}(\sym{A}_1, \sym{B}_1, \sym{pair} (\sym{A}_2, \sym{B}_2, \sym{s}, \sym{t}))
      : \sym{A}_1
    }
  \end{equation}
  where we labeled the equational premises with metavariables~$\upzeta$ and~$\upxi$.
  We may compare~\eqref{eq:example-residue-lhs} with~\eqref{eq:example-residue} to get a matching instantiation
  \begin{equation*}
    I =
       \finmap{
        \sym{A}_1 \mto U_1,
        \sym{A}_2 \mto U_2,
        \sym{B}_1 \mto V_1,
        \sym{B}_2 \mto V_2,
        \sym{s} \mto u,
        \sym{t} \mto v,
        \upzeta \mto \dummy,
        \upxi \mto \dummy
       }
  \end{equation*}
  of its premises over $\Theta; \emptyCtx$.
  Now it would be a mistake to simply instantiate~\eqref{eq:fst-computation} with~$I$ because the equational premises $\upzeta$ and $\upxi$ may not be derivable (the object premises are derivable by \cref{prop:presuppositivity}).
  However, because~\eqref{eq:example-residue-lhs} is object-invertible by \cref{cor:rule-invertible}, proved below, the residual instantiation
  \begin{equation*}
    \residueInst{I} =
       \finmap{
        \sym{A}_1 \mto U_1,
        \sym{A}_2 \mto U_2,
        \sym{B}_1 \mto V_1,
        \sym{B}_2 \mto V_2,
        \sym{s} \mto u,
        \sym{t} \mto v,
        \upzeta \mto \upzeta,
        \upxi \mto \upxi
       },
  \end{equation*}
  is derivable. Hence, we may instantiate~\eqref{eq:fst-computation} with $\residueInst{I}$ to derive
  \begin{align*}
    &\Theta, \zeta \of (U_1 \equiv U_2 \bye \Box), \xi \of (V_1 \equiv V_2 \bye \Box); \emptyCtx \types \\
    &\qquad\qquad\qquad \sym{fst}(U_1, V_1, \sym{pair} (U_2, V_2, u, v)) \equiv u : U_1.
  \end{align*}
  Thus we must still verify
  $
    \Theta; \emptyCtx \types U_1 \equiv U_2
  $
  and
  $
    \Theta; \emptyCtx \types V_1 \equiv V_2
  $,
  in order to conclude~\eqref{eq:example-residue-eq}, precisely as expected.
\end{exa}

% \begin{cor}
%   \label{cor:derivable-empty}%
%   A raw rule $\rawRule{\Xi}{\J}$ is object-invertible if, and only if, the condition from \cref{def:invertible-rule} holds for empty contexts.
% \end{cor}

% \begin{proof}
%   An instantiation $I$ over $\Theta; \Gamma$ may be promoted in the evident way to an instantiation over $(\Theta, \Gamma)$ and $\emptyCtx$, after which~\cref{prop:promotion} applies.
% \end{proof}

Whether a rule is object-invertible depends not just on the rule itself, but on the ambient type theory too, for it may happen that $\Theta; \Gamma \types \act{I}{\J}$ is not derivable by the rule under consideration, but by another one that instantiates to the same conclusion.

\begin{exa}
  Consider the standard type theory whose specific rules are
  \begin{mathpar}
    \infer{ }{\types \isType{\sym{0}}}

    \infer{ }{\types \isType{\sym{1}}}

    \infer{ }{\types \sym{u} : \sym{1}}

    \infer{
      \types \sym{v} : \sym{1}
    }{
      \types \isType{\sym{T}(\sym{v})}
    }

    \infer{
      \types \sym{e} : \sym{0}
    }{
      \types \sym{0} \equiv \sym{1}
    }
  \end{mathpar}
  The derivable object rule
  %
  % \begin{equation*}
  %   \rawRule{
  %     \sym{v} \of (\Box : \sym{0})
  %   }{
  %     \isType{\sym{T}(\sym{v})}
  %   }.
  % \end{equation*}
  %
  \begin{equation*}
    \infer{
       \types \sym{e} : \sym{0}
    }{
      \types \isType{\sym{T}(\sym{e})}
    }
  \end{equation*}
  is not object-invertible, because the instantiation $I = \finmap{\sym{e} \mto \sym{u}}$ yields the derivable judgement $\emptyExt; \emptyCtx \types \isType{\sym{T}(\sym{u})}$, but $\emptyExt; \emptyCtx  \types \sym{u} : \sym{0}$ is not derivable.
  %(Consider the set-theoretic model in which $\sym{0}$ is empty and $\sym{1}$ is a singleton.)
\end{exa}

In the previous example the culprit is the application of term conversion to a metavariable. As it turns out, such conversions of variables are the principal obstruction to invertibility, so we define a syntactic property of judgements which prevents them.

\begin{defi}
  \label{def:natural-for-variables}%
  An object judgement $\Theta; \Gamma \types \JJ$ is \defemph{natural for variables} when the relation $\Theta; \Gamma \natur \JJ$ can be deduced using the rules in \cref{fig:natural-for-variables}.
  \begin{figure}[htpb]
  \centering
  \small
  \begin{ruleframe}
    \begin{mathpar}
      \infer
      {
        \Gamma(\var{a}) = A
      }{
        \Theta; \Gamma \natur \var{a} : A
      }

      \infer
      {
        \var{a} \not\in |\Gamma| \\
        \Theta; \Gamma, \var{a} \of A \natur \JJ[\var{a}/x]
      }{
        \Theta; \Gamma \natur \abstr{x \of A} \JJ
      }

      \infer
      {
        \Theta(\symM) = (\abstr{\vec{x} \of \vec{A}} \; \isType{\Box}) \\\\
        \Theta; \Gamma \natur t_i : A[\upto{\vec{t}}{i}/\upto{\vec{x}}{i}] \quad \text{for $i = 1, \ldots, n$}
      }{
        \Theta; \Gamma \natur \isType{\symM(t_1, \ldots, t_n)}
      }

      \infer
      {
        \Theta(\symM) = (\abstr{\vec{x} \of \vec{A}} \; \Box : B) \\\\
        \Theta; \Gamma \natur t_i : A[\upto{\vec{t}}{i}/\upto{\vec{x}}{i}] \quad \text{for $i = 1, \ldots, n$}
      }{
        \Theta; \Gamma \natur \symM(t_1, \ldots, t_n) : B[\vec{t}/\vec{x}]
      }

      \infer
      {
        \text{Rule for $\symS$ is $\rawRule{\symM_1 \of \BB_1, \ldots, \symM_n \of \BB_n}{\plug{\B}{\symS(\genapp{M}_1, \ldots, \genapp{M}_n)}}$} \\\\
        I = \finmap{\symM_1 \mto e_1, \ldots, \symM_n \mto e_n} \\\\
        \Theta; \Gamma \natur \plug{(\upact{I}{i} \BB_i)}{e_i}  \quad \text{if $\BB_i$ is an object boundary}
      }{
        \Theta; \Gamma \natur \plug{\B'}{\symS(e_1, \ldots, e_n)}
      }
    \end{mathpar}
  \end{ruleframe}
  \caption{Object judgements that are natural for variables}
  \label{fig:natural-for-variables}
\end{figure}
\end{defi}

The point of this definition is that a derivable judgement which is natural for variables has a derivation in which
any application of \rref{TT-Meta} and \rref{TT-Var} is \emph{not} immediately followed by a conversion, unless it appears in a subderivation of an equality judgement. The claim is established by a straightforward induction on the derivation of~$\Theta; \Gamma \natur \JJ$ with the help of \cref{thm:inversion}.

The obvious pattern-matching algorithm scans a pattern and compares it to a term. It instantiates metavariables one by one and possibly out of order, which results in a chain of instantiations, each of which instantiates just one metavariable. Let us study such instantiations.

\begin{defi}
  Let $\Xi = [\symM_1 \of \BB_1, \ldots, \symM_n \of \BB_n]$ be a metavariable context, and $e$ an argument over $\upto{\Xi}{k}$ and the empty context with $\arity{e} = \arity{\BB_k}$. The \defemph{basic instantiation} $\basic{\Xi}{\symM_k}{e}$ is defined by
  \begin{equation}
    \basic{\Xi}{\symM_k}{e}(\symM_i) =
    \begin{cases}
      \genapp{\symM}_i & \text{ if $\symM_k \neq \symM_i$,} \\
      e & \text{ if $\symM_k = \symM_i$.}
    \end{cases}
  \end{equation}
  It is an instantiation of $\Xi$ over the metavariable context
  \begin{equation*}
    \basicCod{\Xi}{\symM_k}{e} =
    [\symM_1 \of \BB'_1, \ldots,\symM_{k-1} \of \BB'_{k-1}, \symM_{k+1} \of \BB'_{k+1}, \ldots, \symM_n \of \BB'_n]
  \end{equation*}
  and the empty context, where
  $\BB'_j = \upact{\basic{\Xi}{\symM_k}{e}}{j} \BB_j$.
\end{defi}

\begin{lem}
  \label{lm:basic-inst-derivable}
  A basic instantiation $\basic{\Xi}{\symM_k}{e}$ is derivable if $\types \isExt{\Xi}$ and $\upto{\Xi}{k} \types \plug{\BB_k}{e}$, in which case $\types \isExt{\basicCod{\Xi}{\symM_k}{e}}$ also holds.
\end{lem}

\begin{proof}
  For $i < k$, the judgement
  $
  \basicCod{\Xi}{\symM_k}{e} \types \plug{(\upact{\basic{\Xi}{\symM_k}{e}}{i} \BB_i)}{\genapp{\symM}_i}
  $
  holds by abstraction and the metavariable rule, where we invert $\types \isExt{\Xi}$ to validate the abstractions.

  The judgement
  $\basicCod{\Xi}{\symM_k}{e} \types \plug{(\upact{\basic{\Xi}{\symM_k}{e}}{k} \BB_k)}{e}$
  follows by weakening from $\upto{\Xi}{k} \types \plug{\BB_k}{e}$ because $\upto{\basicCod{\Xi}{\symM_k}{e}}{k} = \upto{\Xi}{k}$.

  For $i > k$, we again use abstraction and the metavariable rule, where abstractions are now validated by inversion of $\types \isExt{\Xi}$ and \cref{thm:admiss-of-inst} applied to~$\upto{\basic{\Xi}{\symM_k}{e}}{i}$.

  The derivation of $\types \isExt{\basicCod{\Xi}{\symM_k}{e}}$ has two parts. First,
  $\upto{\basicCod{\Xi}{\symM_k}{e}}{k}$ coincides with $\upto{\Xi}{k}$
  and so we just reuse $\types \isExt{\upto{\Xi}{k}}$.
  For $i > k$, we derive $\upto{\basicCod{\Xi}{\symM_k}{e}}{\symM_i} \types \BB'_i$
  as the instantiation of $\upto{\Xi}{i} \types \BB_i$ by
  \begin{equation*}
    \Inst
    {\upto{\basic{\Xi}{\symM_k}{e}}{i}}
    {\upto{\Xi}{i}}
    {\upto{\basicCod{\Xi}{\symM_k}{e}}{\symM_i}}
    {\emptyCtx},
  \end{equation*}
  which is observed to be derivable.
\end{proof}

We define particular compositions of chains of basic instantiations, as follows.
Given a metavariable context $\Xi = [\symM_1 \of \BB_1, \ldots, \symM_n \of \BB_n]$ and an instantiation
\begin{equation*}
  I = \finmap{\symM_1 \mto e_1, \ldots, \symM_n \mto e_n}
\end{equation*}
of~$\Xi$ over~$\Theta; \emptyCtx$, define the instantiation
\begin{equation*}
 \Inst
   {\basicComp[\Xi,\Theta,I]{\symM_{i_1}, \ldots, \symM_{i_m}}}
   {\Xi}
   {\basicCompCod[\Xi,\Theta,I]{\symM_{i_1}, \ldots, \symM_{i_m}}}{\emptyExt}
\end{equation*}
and the metavariable context $\basicCompCod[\Xi,\Theta,I]{\symM_{i_1}, \ldots, \symM_{i_m}}$ by
\begin{align*}
  \basicCompCod[\Xi,\Theta,I]{} &= \finmap{\Theta, \Xi} \\
  \basicComp[\Xi,\Theta,I]{} &= \finmap{\symM_1 \mto \genapp{\symM}_1, \ldots, \symM_n \mto \genapp{\symM}_n} \\
  \basicCompCod[\Xi,\Theta,I]{\symM_{i_1}, \ldots, \symM_{i_{m+1}}} &= \basicCod{\basicCompCod[\Xi,\Theta,I]{\symM_{i_1}, \ldots, \symM_{i_m}}}{\symM_{i_{m+1}}}{e_{i_{m+1}}} \\
  \basicComp[\Xi,\Theta,I]{\symM_{i_1}, \ldots, \symM_{i_{m+1}}}
  &=
    \begin{aligned}[t]
    \basic{\basicCompCod[\Xi,\Theta,I]{\symM_{i_1}, \ldots, \symM_{i_m}}}{\symM_{i_{m+1}}}{e_{i_{m+1}}} \circ {} & \\ \basicComp[\Xi,\Theta,I]{\symM_{i_1}, \ldots, \symM_{i_m}} &
    \end{aligned}
\end{align*}
In the above definition we require $|\Xi| \cap |\Theta| = \emptyset$ and that $\symM_{i_1}, \ldots, \symM_{i_m}$ are all distinct.
We elide the subscripts and write $\basicComp{\symM_{i_1}, \ldots, \symM_{i_m}}$ and $\basicCompCod{\symM_{i_1}, \ldots, \symM_{i_m}}$ when no confusion can arise.
A straightforward induction shows that
\begin{equation*}
  \basicCompCod[\Xi,\Theta,I]{\symM_{i_1}, \ldots, \symM_{i_m}}(\symM_j) =
  \act{\basicComp[\Xi,\Theta,I]{\symM_{i_1}, \ldots, \symM_{i_m}}} \BB_j
\end{equation*}
for any $\symM_j \in |\Xi| \setminus \set{\symM_{i_1}, \ldots, \symM_{i_m}}$.
The instantiation $\basicComp[\Xi,\Theta,I]{\symM_{i_1}, \ldots, \symM_{i_m}}$ plays a role in proving object-invertibility, because $\set{\symM_{i_1}, \ldots, \symM_{i_m}} = \objmv{\Xi}$ implies
\begin{equation*}
  \basicComp[\Xi,\Theta,I]{\symM_{i_1}, \ldots, \symM_{i_m}} = \residueInst{I}
  \quad \text{and} \quad
  \basicCompCod[\Xi,\Theta,I]{\symM_{i_1}, \ldots, \symM_{i_m}} = \residue{\Xi}{I}.
\end{equation*}

We are now in possession of all the ingridients necessary to relate pattern matching and object-invertibility.
Recall that, given a rule $\rawRule{\Xi}{\plug{\B}{p}}$ whose head is a pattern~$p$ and an instantiation $I$ of~$\Xi$ over $\Theta; \emptyCtx$, we would like to show that $\residueInst{I}$ is derivable. The following lemma, whose purpose we explain shortly, is a stepping stone to accomplishing the goal.

\begin{lem}
  \label{lem:derivable-extension}
  In a standard type theory, let $\rawRule{\Xi}{\plug{\B}{p}}$ be a derivable object rule which is natural for variables, $p$ a pattern, and~$I$ an instantiation of~$\Xi = [\symM_1 \of \BB_1, \ldots, \symM_n \of \BB_n]$ over~$\Theta; \emptyCtx$ such that $|\Xi| \cap |\Theta| = \emptyset$, and $\Theta; \emptyCtx \types \act{I} (\plug{\B}{p})$ is derivable.

  Suppose $\vec{N} = (N_1, \ldots, N_m)$ is a sequence of distinct metavariables,
  $\set{N_1, \ldots, N_m} \subseteq |\Xi|$,
  $\mv{\B} \subseteq \set{N_1, \ldots, N_m} \cup \mv{p}$,
  and both
  $\types \isExt{\basicCompCod[\Theta,\Xi, I]{\vec{N}}}$ and
  $\basicComp[\Theta,\Xi, I]{\vec{N}}$ are derivable.
  Then $\vec{N}$ can be extended to a sequence of distinct metavariables
  $\vec{N}' = (N_1, \ldots, N_\ell)$ such that
  $\set{N_1, \ldots, N_\ell}  =  \set{N_1, \ldots, N_m} \cup \mv{p}$, and both
  $\types \isExt{\basicCompCod{\vec{N}'}}$ and
  $\basicComp{\vec{N}'}$ are derivable.
\end{lem}

\noindent
Let us explain how the lemma shall be used. As noted above, $\residueInst{I}$ coincides with $\basicComp{\vec{N}}$ when $\vec{N}$ lists all of $\objmv{\Xi}$.
Therefore, we may establish derivability of~$\residueInst{I}$ by starting with $\vec{N} = ()$ and extending it with metavariables until it encompasses~$\objmv{\Xi}$.
The lemma guarantees that one such extension step can be done with the aid of patterns in a way that preserves derivability of~$\basicComp{\vec{N}}$.

\begin{proof}[Proof of \cref{lem:derivable-extension}]
  Let $I = \finmap{\symM_1 \mto e_1, \ldots, \symM_n \mto e_n}$.
  We proceed by induction on the structure of~$p$, and elide the subscripts to keep the notation shorter.

  \inCaseText{$p = \symM_k$, $\BB_k = (\isType{\Box})$, and $\B = (\isType{\Box})$}
  If $\symM_k$ appears in $\vec{N}$ we let $\ell = m$ and we are done. Otherwise we set $\ell = m+1$ and $N_{m+1} = \symM_k$.
  Because composition of derivable instantiations is derivable, we only need to show that $\basic{\basicCompCod{\vec{M}}}{\symM_k}{e_k}$ is derivable,
  which by \cref{lm:basic-inst-derivable} reduces to
  \begin{equation*}
    \upto{\basicCompCod{\vec{N}}}{\symM_k} ; \emptyCtx \types
     \plug
     {(\act{\basicComp{\vec{N}}} \BB_{k})}
     {e_{k}},
  \end{equation*}
  which equals
  $$\upto{\basicCompCod{\vec{N}}}{\symM_k} ; \emptyCtx \types \isType{e_{k}}.$$
  It is derivable by weakening from the assumption $\Theta ; \emptyCtx \types \isType{e_k}$.

  \inCaseText{$p = \symS(q_1, \ldots, q_m)$}
  Suppose the symbol rule for~$\symS$ is
  \begin{equation*}
    \rawRule
    {\symM'_1 \of \BB'_1, \ldots, \symM'_j \of \BB'_j}
    {\plug{\B'}{\symS(\genapp{\symM'}_1, \ldots, \genapp{\symM'}_j)}}.
  \end{equation*}
  By applying \cref{cor:args-derivable} to $\Xi; \emptyCtx \types \plug{\B}{\symS(\vec{q})}$ and letting
  $K = [\symM'_1 \mto q_1, \ldots, \symM'_j \mto q_j]$, we obtain for $i = 1, \ldots, j$ derivations of
  \begin{equation}
    \label{eq:derinst-1}
    \Xi; \emptyCtx \types \plug{(\upact{K}{i} \BB'_i)}{q_i}.
  \end{equation}
  Similarly, from derivability of $\Theta; \emptyCtx \types \plug{(\act{I} \B)}{\symS(\act{I} \vec{q})}$ we obtain derivability of
  \begin{equation}
    \label{eq:derinst-3}
    \Theta; \emptyCtx \types \plug{(\upact{(\act{I} K)}{i} \BB'_i)}{\act{I} q_i},
  \end{equation}
  which is equal to
  \begin{equation}
    \label{eq:derinst-2}
    \Theta; \emptyCtx \types \act{I} (\plug{(\upact{K}{i} \BB'_i)}{q_i}).
  \end{equation}
  We define $\vec{L}_0, \ldots, \vec{L}_j$ such that
  $\vec{L}_0 = \vec{N}$, and
  for $i = 1, \ldots, j$, the sequence $\vec{L}_i$ extends $\vec{L}_{i-1}$ by $\mv{q_i}$, and both $\types \isExt{\basicCompCod{\vec{L}_i}}$ and $\basicComp{\vec{L}_i}$ are derivable.
  We may then finish the proof by taking $\vec{N}' = \vec{L}_j$.
  Assuming $\vec{L}_{i-1}$ has been constructed, we consider two cases.

  First, if $q_i$ is a non-abstracted object pattern then we obtain $\vec{L}_i$ by applying the induction hypothesis to~\eqref{eq:derinst-1},~\eqref{eq:derinst-2} and~$\vec{L}_{i-1}$. We may do so because $\mv{\upact{K}{i} \BB'_i} \subseteq \mv{q_1} \cup \cdots \cup \mv{q_{i-1}}$, which is contained in~$\vec{L}_{i-1}$.

  Second, if $q_i = \abstr{\vec{x}} \symM_k(\vec{x})$ we proceed as follows. If $\symM_k$ appears in $\vec{L}_{i-1}$, we take $\vec{L}_i = \vec{L}_{i-1}$ and we are done.
  Otherwise, we take $\vec{L}_i = (\vec{L}_{i-1}, \symM_k)$.
  We need to show derivability of $\types \isExt{\basicCompCod{\vec{L}_i}}$ and $\basicComp{\vec{L}_i}$. Because
  $
  \basicComp{\vec{L}_i} =
    \basic{\basicCompCod{\vec{L}_{i-1}}}{\symM_k}{e_k}
    \circ
    \basicComp{\vec{L}_{i-1}}
  $
  and $\basicComp{\vec{L}_{i-1}}$ is derivable it suffices to show
  that $\basic{\basicCompCod{\vec{L}_{i-1}}}{\symM_k}{e_k}$ is derivable,
  and therefore by \cref{lm:basic-inst-derivable} that
  \begin{equation}
    \label{eq:derivable-2}%
    \upto{\basicCompCod{\vec{L}_{i-1}}}{\symM_k} ; \emptyCtx
    \types
    \plug
    {(\act{\basicComp{\vec{L}_{i-1}}} \BB_k)}
    {e_k}.
  \end{equation}
  We claim that~\eqref{eq:derivable-2} is just a weakening of~\eqref{eq:derinst-3}.
  Obviously, $\upto{\basicCompCod{\vec{L}_{i-1}}}{\symM_k}$ extends $\Theta$ and $\act{I} q_i = e_k$.
  It remains to be seen that $\act{\basicComp{\vec{L}_{i-1}}} \BB_k$ and
  $\upact{(\act{I} K)}{i} \BB'_i$ are the same.
  The judgement~\eqref{eq:derinst-1} equals $\Xi; \emptyCtx \types \plug{(\upact{K}{i} \BB'_i)}{\abstr{\vec{x}} \symM_k(\vec{x})}$. By the naturality-for-variables assumption it is derivable without conversions, which is only possible if $\upact{K}{i} \BB'_i$ is $\BB_k$.
  Therefore,
  \begin{equation*}
    \upact{(\act{I} K)}{i} \BB'_i =
    \act{I} (\upact{K}{i} \BB'_i) =
    \act{\basicComp{\vec{L}_{i-1}}} (\upact{K}{i} \BB'_i) =
    \act{\basicComp{\vec{L}_{i-1}}} \BB_k,
  \end{equation*}
  where the second step is valid because $\mv{\upact{K}{i} \BB'_i} \subseteq \mv{q_1} \cup \cdots \cup \mv{q_{i-1}}$, which is contained in~$\vec{L}_{i-1}$.
\end{proof}

\begin{cor}
  \label{cor:rule-invertible}%
  In a standard type theory, consider a derivable finitary object rule $\rawRule{\Xi}{\plug{\B}{p}}$ which is natural for variables. If $p$ is a pattern and $\mv{p} = \objmv{\Xi}$ then the rule is object-invertible.
\end{cor}
\begin{proof}
  Consider an instantiation~$I$ of~$\Xi$ over~$\Theta; \emptyCtx$, such that
  $\types \isExt{\Theta}$ and $\Theta; \emptyCtx \types \plug{(\act{I} \B)}{\act{I} p}$ are derivable. Without loss of generality we may assume $|\Xi| \cap |\Theta| = \emptyset$.

  We apply \cref{lem:derivable-extension} with the empty sequence $\vec{N} = ()$, noting that
  $\mv{\B} \subseteq \mv{p}$, that $\basicCompCod{} = \finmap{\Theta, \Xi}$ and that $\types \isExt{\finmap{\Theta, \Xi}}$ is derivable because the rule is finitary and we assumed $\types \isExt{\Theta}$.
  This way we obtain a sequence $\vec{N}' = (N'_1, \ldots, N'_\ell)$ such that $\mv{p} = \set{N'_1, \ldots, N'_\ell}$ and $\basicComp{\vec{N}'}$ is derivable.
  Because $\mv{p} = \objmv{\Xi}$, it follows that $\basicComp{\vec{N}'}$ coincides with $\residueInst{I}$, hence it is derivable too.
\end{proof}

\subsection{Computation and extensionality rules}
\label{sec:extens-comp-rules}

The equality checking algorithm uses two kinds of equational rules, which we describe here and prove that they have the desired properties. First, we have the rules that govern normalization.

\begin{defi}
  \label{def:type-computation-rule}%
  A derivable finitary rule $\rawRule{\Theta}{A \equiv B}$ is a \defemph{type computation rule}
  if $\rawRule{\Theta}{\isType{A}}$ is deterministic and object-invertible.
\end{defi}

\begin{defi}
  \label{def:term-computation-rule}%
  A derivable finitary rule $\rawRule{\Theta}{u \equiv v : A}$ is a \defemph{term computation rule}
  if $u$ is a term symbol application and $\rawRule{\Theta}{u : \natty{\Theta; \emptyCtx}{u}}$ is deterministic and object-invertible.
\end{defi}

The reason behind the first condition in the definition of a term computation rule is that for term symbol applications~\cref{prop:natty-inst} holds, which is needed in the proof of soundness (\cref{thm:normalization-sound}). We exhibit in~\cref{ex:ext-rules-no-equations} what can go wrong if we allow for a metavariable as the lefthand-side of the equation.
One might hope that the second condition in~\cref{def:term-computation-rule} could be relaxed to $\rawRule{\Theta}{u : A}$. However, the additional flexibility is only apparent, for if a term has a type then it has the natural type as well. In any case, in the proofs of soundness (\cref{thm:normalization-sound,thm:checking-sound}) we rely on having the natural type.

A computation rule may be recognized using the following criterion.

\begin{prop}
  \label{prop:comp-rule-suff-cond}%
  In a standard type theory:
  \begin{enumerate}
  \item
    A derivable finitary rule $\rawRule{\Xi}{P \equiv B}$ is a type computation rule if $P$ is a type pattern, $\mv{P} = \objmv{\Xi}$, and $\rawRule{\Xi}{\isType{P}}$ is natural for variables.

  \item
    A derivable finitary rule $\rawRule{\Xi}{p \equiv v : A}$ is a term computation rule if~$p$ is a term pattern, $\mv{p} = \objmv{\Xi}$, and $\rawRule{\Xi}{p : \natty{\Xi; \emptyCtx}{p}}$ is natural for variables.
  \end{enumerate}

\end{prop}

\begin{proof}
  The rule $\rawRule{\Xi}{\isType{P}}$ is derivable by \cref{prop:presuppositivity}, and $\rawRule{\Xi}{p : \natty{\Xi; \emptyCtx}{p}}$ by
  \cref{prop:presuppositivity} and \cref{cor:natty-derivable}. Observe that $\mv{P} = \objmv{\Xi}$ and $\mv{p} = \objmv{\Xi}$,
  and apply \cref{prop:rule-deterministic} and \cref{cor:rule-invertible} respectively to $\rawRule{\Xi}{\isType{P}}$ and to $\rawRule{\Xi}{p : \natty{\Xi; \emptyCtx}{p}}$.
\end{proof}

\begin{exa}
  \label{ex:computation-rule-prod}%
  Typical $\beta$-rules satisfy the conditions of \cref{prop:comp-rule-suff-cond}, after their left-hand sides have been linearized, as in \cref{exa:linearize}.
  Another example is the $\beta$-rule for application
  \begin{equation*}
    \infer{
      \types \isType{\sym{A}}
      \\
      \types \abstr{x \of \sym{A}} \; \isType{\sym{B}}
      \\
      \types \abstr{x \of \sym{A}} \; \sym{s} : \sym{B}(x)
      \\
      \types \sym{t} : \sym{A}
    }{
      \types
      \sym{apply}(
         \sym{A},
         \abstr{x}\sym{B}(x),
         \uplambda (\sym{A}, \abstr{x}\sym{B}(x), \abstr{x} \sym{s}(x)),
         \sym{t}
      )
      \equiv
      \sym{s}(\sym{t})
      : \sym{B}(\sym{t})
    }
  \end{equation*}
  whose linearized form is
  \begin{equation*}
    \infer{
      \types \isType{\sym{A}_1}
      \\
      \types \abstr{x \of \sym{A}_1} \; \isType{\sym{B}_1}
      \\\\
      \types \isType{\sym{A}_2}
      \\
      \types \abstr{x \of \sym{A}_2} \; \isType{\sym{B}_2}
      \\\\
      \types \abstr{x \of \sym{A}_2} \; \sym{s} : \sym{B_2}(x)
      \\
      \types \sym{t} : \sym{A}_1
      \\\\
      \types \sym{A}_1 \equiv \sym{A}_2 \\
      \types \abstr{x \of \sym{A}_1} \sym{B}_1(x) \equiv \sym{B}_2(x)
    }{
      \types
      \sym{apply}(
         \sym{A}_1,
         \abstr{x}\sym{B}_1(x),
         \uplambda (\sym{A}_2, \abstr{x}\sym{B}_2(x), \abstr{x} \sym{s}(x)),
         \sym{t}
      )
      \equiv
      \sym{s}(\sym{t})
      : \sym{B}_1(\sym{t})
    }
  \end{equation*}
  which satisfies \cref{prop:comp-rule-suff-cond}.

  We also allow the somewhat unusual rule
  \begin{equation*}
    \infer{ }{
      \types \isType{\sym{U}}
    }
    \qquad \qquad
    \infer{
      \types \isType{\sym{A}}
    }{
      \types \sym{A} \equiv \sym{U}
    }
  \end{equation*}
  because it allows us to dispense with all questions about equality of types in case we want to work with an uni-typed theory (some would call it untyped).
\end{exa}

The second kind of rules is used by the algorithm to reduce an equation to subordinate equations by matching its type.

\begin{defi}
  \label{def:extensionality-rule}
  An \defemph{extensionality rule} is a derivable finitary rule of the form
  \begin{equation*}
    \rawRule{\Theta, \sym{s} \of (\Box : A), \sym{t} \of (\Box : A), \Phi}{\sym{s} \equiv \sym{t} : A}
  \end{equation*}
  such that
  $\Phi$ contains only equational premises, and
  $\rawRule{\Theta}{\isType{A}}$ is deterministic and object-invertible.
\end{defi}

An extensional rule may be recognized with the following criterion.

\begin{prop}
  \label{prop:extensionality-rule-criterion}%
  In a standard type theory, a derivable finitary rule of the form
  \begin{equation*}
    \rawRule{\Xi, \sym{s} \of (\Box : P), \sym{t} \of (\Box : P), \Phi}{\sym{s} \equiv \sym{t} : P}
  \end{equation*}
  is an extensionality rule if $\Phi$ contains only equational premises, $P$ is a type pattern, $\mv{P} = \objmv{\Xi}$, and $\rawRule{\Xi}{\isType{P}}$ is natural for variables.
\end{prop}

\begin{proof}
  Apply \cref{prop:rule-deterministic} and \cref{cor:rule-invertible} to $\rawRule{\Xi}{\isType{P}}$.
\end{proof}

Extensionality rules that one finds in practice typically satisfy the above syntactic condition, even without linearization. Here are a few.

\begin{exa}
  \label{ex:extesionality-product}%
  \label{ex:dep-prod-ext}%
  Extensionality rules typically state that elements of a type are equal when their parts are equal. For example, extensionality for simple products states that pairs are equal if their components are equal:
  \begin{equation}
    \label{eq:ext-simple-prod}
    \infer{
        \types \isType{\sym{A}}
        \\
        \types \isType{\sym{B}}
        \\
        \types \sym{s} : \sym{A} \times \sym{B}
        \\
        \types \sym{t} : \sym{A} \times \sym{B}
        \\\\
        \types \mathsf{fst}(\sym{A}, \sym{B}, \sym{s}) \equiv \mathsf{fst}(\sym{A}, \sym{B}, \sym{t}) : \sym{A}
        \\
        \types \mathsf{snd}(\sym{A}, \sym{B}, \sym{s}) \equiv \mathsf{snd}(\sym{A}, \sym{B}, \sym{t}) : \sym{B}
    }{
      \types \sym{s} \equiv \sym{t} : \sym{A} \times \sym{B}
    }
  \end{equation}
  Similarly, the extensionality rule for dependent functions states that they are equal if their generic applications are equal:
  \begin{equation*}
    \infer{
        \types \isType{\sym{A}}
        \\
        \types \abstr{x \of \sym{A}} \; \isType{\sym{B}}
        \\
        \types \sym{s} : \Uppi(\sym{A}, \abstr{x} \sym{B}(x))
        \\
        \types \sym{t} : \Uppi(\sym{A}, \abstr{x} \sym{B}(x))
        \\\\
        \types
           \abstr{x \of \sym{A}} \;
           \sym{apply}(\sym{A}, \abstr{x} \sym{B}(x), \sym{s}, \sym{x})
           \equiv
           \sym{apply}(\sym{A}, \abstr{x} \sym{B}(x), \sym{t}, \sym{x})
           : \sym{B}(x)
    }{
      \types \sym{s} \equiv \sym{t} : \Uppi(\sym{A}, \abstr{x} \sym{B}(x))
    }
  \end{equation*}
  The above is not to be confused with \emph{propositional} function extensionality, which is a certain term that maps point-wise propositional equality of functions to their propositional equality.
\end{exa}

\begin{exa}
  \label{ex:ext-rules-no-equations}
  Some extensionality rules have no equational premises.
  The first one that comes to mind is the rule stating that all elements of the unit type are equal:
  \begin{equation*}
    \infer{
      \types \sym{s} : \sym{unit} \\
      \types \sym{t} : \sym{unit}
    }{
      \types \sym{s} \equiv \sym{t} : \sym{unit}
    }
  \end{equation*}
  The corresponding $\eta$-rule ($\star$ is the canonical inhabitant of $\sym{unit}$)
  \begin{equation*}
    \infer{
      \types \sym{t} : \sym{unit}
    }{
      \types \sym{t} \equiv \star
    }
  \end{equation*}
  cannot be incorporated as a computation rule naively because the bare metavariable on the left-hand side matches any term, even if its type is not (judgementally equal to) unit. Since our normalization procedure in~\cref{sec:principal-and-normalization} does not check for equality of types separately, such rules do not behave well as computation rules.
  Another rule of this kind is the judgemental variant of Uniqueness of Identity Proofs (UIP) which equates any two proofs of a propositional identity:
  \begin{equation*}
    \infer{
      \types \isType{\sym{A}} \\
      \types \sym{a} : \sym{A} \\
      \types \sym{b} : \sym{A} \\
      \types \sym{p} : \sym{Id}(\sym{A}, \sym{a}, \sym{b}) \\
      \types \sym{q} : \sym{Id}(\sym{A}, \sym{a}, \sym{b})
    }{
      \types \sym{p} \equiv \sym{q} : \sym{Id}(\sym{A}, \sym{a}, \sym{b})
    }
  \end{equation*}
  The corresponding $\eta$-rule is as troublesome as the one for $\sym{unit}$:
  \begin{equation*}
    \infer{
      \types \isType{\sym{A}} \\
      \types \sym{a} : \sym{A} \\
      \types \sym{p} : \sym{Id}(\sym{A}, \sym{a}, \sym{a})
    }{
      \types \sym{p} \equiv \sym{refl}(\sym{A}, \sym{a}) : \sym{Id}(\sym{A}, \sym{a}, \sym{a})
    }
  \end{equation*}
  The principle has been used, for example, in the cubical type theory XTT for Bishop sets~\cite{sterling20:_cubic_languag_bishop_sets}.

  Here is one last example:
  \begin{mathpar}
    \infer{
      \types \isType{\sym{A}}
    }{
      \types \isType{\|\sym{A}\|}
    }

    \infer{
      \types \isType{\sym{A}} \\
      \types \sym{t} : \sym{A}
    }{
      \types |\sym{t}| : \|\sym{A}\|
    }

    \infer{
      \types \isType{\sym{A}} \\
      \types \sym{u} : \|\sym{A}\| \\
      \types \sym{v} : \|\sym{A}\|
    }{
      \types \sym{u} \equiv \sym{v} : \|\sym{A}\|
    }
  \end{mathpar}
  The above rules describe a kind of ``judgemental truncation'', which is like the propositional truncation from homotopy type theory, except that it equates all terms of $\|\sym{A}\|$ judgementally. It is unclear what elimination rule of judgemental truncation would be, but one is reminded of the \emph{proof-irrelevant propositions}~\cite{gilbert19:_defin_k}.
\end{exa}

%%% Local Variables:
%%% mode: latex
%%% TeX-master: "equality-checking"
%%% End:

% chktex-file 46
% chktex-file 45
% chktex-file 1
% chktex-file 26
% chktex-file 24
\subsection{Principal arguments and normalization}
\label{sec:principal-and-normalization}

Normalization rewrites an expression $\symS(e_1, \ldots, e_n)$ by normalizing some of the arguments $e_1, \ldots, e_n$, applying a computation rule, and repeating the process.
We say that an argument~$e_i$ (or more precisely, its position~$i$) is \emph{principal} for~$\symS$ if it is so normalized.
By varying the selection of principal arguments we may control the algorithm to compute various kinds of normal form. For example, in $\lambda$-calculus the weak-head normal form is obtained when the only principal argument is the head of an application, while taking all arguments to be principal leads to the strong normal form.
Our algorithm is flexible in this regard, as it takes the information about principality of arguments as input. In \cref{sec:choose-determ-princ} we discuss how appropriate principal arguments can be chosen.

In specific cases normal forms are characterized by their syntactic structure, for example a normal form in the $\lambda$-calculus is an expression without $\beta$-redeces. One then proves that the normalization procedure always leads to a normal form. We are faced with a general situation in which no such syntactic characterization is available. Luckily, the algorithm never needs to recognize normal forms, although we do have to keep track of which expressions have already been subjected to the normalization procedure, so that we avoid normalizing them again.

\defemph{Normalization} is parametrized by the following data:
\begin{enumerate}
\item a standard type theory~$T$,
\item a family $\mathcal{C}$ of computation rules for~$T$ (\cref{def:type-computation-rule,def:term-computation-rule}),
\item for each symbol~$\symS$ taking~$k$ arguments, a set $\principal{\symS} \subseteq \set{1, \ldots, k}$ of its \defemph{principal arguments},
\end{enumerate}
It has three interdependent variations:
\begin{align*}
  \Theta; \Gamma &\types \plug{\BB}{e \compute e'}
  &&\text{normalize argument $e$ to $e'$,}
  \\
  \Theta; \Gamma &\types \plug{\B}{\symS(\vec{e}) \computeP \symS(\vec{e}')}
  &&\text{normalize the principal arguments of~$\symS$,}
  \\
  \Theta; \Gamma &\types \plug{\B}{e \computeC e'}
  &&\text{use a computation rule to rewrite~$e$ to~$e'$.}
\end{align*}
Specifically,
\begin{equation*}
  \Theta; \Gamma \types \isType{(A \compute A')}
  \qquad\text{and}\qquad
  \Theta; \Gamma \types t \compute t' : B
\end{equation*}
respectively express the facts that the type~$A$ normalizes to~$A'$ and the term~$t$ to~$t'$.
\Cref{fig:normalization} specifies the normalization procedure. Note that normalization is mutually recursive with equality checking, because the rule for $\computeC$ resolves equational premises using equational checking from \cref{fig:equality-checking}.
We omitted the clauses for metavariable applications, as they are analogous to symbol applications. That is, for the purposes of normalization and equality checking, an object metavariable~$\symM$ with boundary~$\BB$ and arity $\arity{\BB} = (c, n)$ is construed as a primitive symbol of syntactic class~$c$ taking~$n$ term arguments.

Normalization of arguments is syntax-directed and deterministic, and so is normalization of principal arguments. However, the applications of computation rules need not terminate, and the computation rules may be a source of non-determinism when several apply to the same expression. We discuss strategies for dealing with these issues in \cref{sec:deter-term-complet}.

\begin{figure}[htbp]
  \begin{center}
  \small
  \begin{ruleframe}
    \begin{mathpar}
      \infer{
        (\rawRule{\Xi}{\plug{\B'}{p \equiv v}}) \in \mathcal{C}
        \\
        \matches{\Xi}{p}{s}{I}
        \\\\
        \Theta; \Gamma \types \act{I} (\plug{\BB}{e \cmp e'})
        \quad\text{for $(\symM \of \plug{\BB}{e \equiv e' \bye \Box}) \in \Xi$}
      }{
        \Theta; \Gamma \types \plug{\B}{s \computeC \act{I} v}
      }

      \infer{
        {\begin{aligned}
        &\text{Rule for $\symS$ is $\rawRule{\symM_1 \of \BB_1, \ldots, \symM_n \of \BB_n}{\plug{\B'}{\symS(\genapp{\symM}_1, \ldots, \genapp{\symM}_n)}}$}
        \\
        &{\begin{aligned}
           \Theta; \Gamma \types \plug{(\finmap{\symM_1 \mto e_1, \ldots, \symM_{i-1} \mto e_{i-1}} \BB_i)}{e_i \compute e'_i}
           &\quad\text{if $i \in \principal{\symS}$} \\
           e_i = e'_i
           &\quad\text{if $i \notin \principal{\symS}$}
         \end{aligned}}
        \end{aligned}}
      }{
        \Theta; \Gamma \types \plug{\B}{\symS(\vec{e}) \computeP \symS(\vec{e}')}
      }

      \infer{
        \Theta; \Gamma \types \plug{\B}{\symS(\vec{e}) \computeP \symS(\vec{e}')} \\
        \Theta; \Gamma \types \plug{\B}{\symS(\vec{e}') \computeC e''} \\
        \Theta; \Gamma \types \plug{\B}{e'' \compute e'''}
      }{
        \Theta; \Gamma \types \plug{\B}{\symS(\vec{e}) \compute e'''}
      }

      \infer{
        \Theta; \Gamma \types \plug{\B}{\symS(\vec{e}) \computeP \symS(\vec{e}')}
        \\\\
        \notmatches{\Xi}{p}{\symS(\vec{e}')}
        \quad\text{for $(\rawRule{\Xi}{\plug{\B'}{p \equiv v}}) \in \mathcal{C}$}
      }{
        \Theta; \Gamma \types \plug{\B}{\symS(\vec{e}) \compute \symS(\vec{e}')}
      }

      \infer{
        \var{a} \notin |\Gamma|
        \\
        \Theta; \Gamma, \var{a} : A \types \plug{(\BB[\var{a}/x])}{e[\var{a}/x] \compute e'}
      }{
        \Theta; \Gamma \types \abstr{x \of A} \; \plug{\BB}{\abstr{x} e \compute \abstr{x} e'[x/\var{a}]}
      }
      \\
      \infer{
      }{
        \Theta; \Gamma \types \var{a} \compute \var{a} : A
      }

      \infer{
      }{
        \Theta; \Gamma \types \plug{\B}{\dummy \compute \dummy}
      }
    \end{mathpar}
  \end{ruleframe}
  \end{center}
  \caption{Normalization with computation rules $\mathcal{C}$ and principal arguments~$\wp$.}
  \label{fig:normalization}
\end{figure}

%%% Local Variables:
%%% mode: latex
%%% TeX-master: "equality-checking"
%%% End:

% chktex-file 46
% chktex-file 45
% chktex-file 1
% chktex-file 26
% chktex-file 24
\subsection{Type-directed and normalization phases}
\label{sec:type-directed-equality-checking}

We are finally ready to describe equality checking, which is performed by several mutually recursive phases:
\begin{align*}
  \Theta; \Gamma &\types \plug{\BB}{e \cmp e'}
  &&\text{$e$ and $e'$ are equal arguments}
  \\
  \Theta; \Gamma &\types s \cmpE t : A
  &&\text{$s$ and $t$ are extensionally equal}
  \\
  \Theta; \Gamma &\types s \cmpN t : A
  &&\text{normalized terms $s$ and $t$ are equal}
  \\
  \Theta; \Gamma &\types A \cmpN B
  &&\text{normalized types $A$ and $B$ are equal}
\end{align*}
The first one is the general comparison of arguments~$e$ and $e'$ of an object boundary~$\BB$,
the second one the \defemph{type-directed phase} which applies extensionality rules by matching the type,
and the third the \defemph{normalization phase} which compares normalized expressions.
We review the inductive clauses specifying these, shown in \cref{fig:equality-checking}.
They are parametrized by a standard type theory $T$, a family of extensionality rules $\mathcal{E}$ over $T$, a family of computation rules~$\mathcal{C}$ over $T$, and a specification of principal arguments~$\wp$.
We again treat metavariables as primitive symbols.

\begin{figure}[htbp]
  \begin{center}
    \small
    \begin{ruleframe}
    \begin{mathpar}
      \infer{
        \Theta; \Gamma \types \isType{(A \compute A')} \\
        \Theta; \Gamma \types u \cmpE v : A'
      }{
        \Theta; \Gamma \types u \cmp v : A
      }

      \infer{
        \Theta; \Gamma \types \isType{(A \compute A')} \\
        \Theta; \Gamma \types \isType{(B \compute B')} \\
        \Theta; \Gamma \types \isType{(A' \cmpN B')}
      }{
        \Theta; \Gamma \types \isType{(A \cmp B)}
      }

      \infer{
        \var{a} \notin |\Gamma| \\
        \Theta; \Gamma, \var{a} \of A \types \plug{(\BB[\var{a}/x])}{e[\var{a}/x] \sim e'[\var{a}/x]}
      }{
        \Theta; \Gamma \types \abstr{x \of A} \; \plug{\BB}{\abstr{x}e \sim \abstr{x}e'}
      }

      \infer{
        (\rawRule{\Xi, \sym{s} \of (\Box : P), \sym{t} \of (\Box : P), \Phi}{\sym{s} \equiv \sym{t} : P})
        \in \mathcal{E}
        \\
        \matches{\Xi}{P}{A}{I}
        \\\\
        {
          \begin{aligned}
            \Theta; \Gamma &\types \act{I} (\plug{\BB}{e \cmp e'})
            &&\text{for $(\symM \of \plug{\BB}{e \equiv e' \bye \Box}) \in \Xi$}
            \\
            \Theta; \Gamma &\types \act{\finmap{I, \sym{s} \mto u, \sym{t} \mto v}} (\plug{\BB}{e \cmp e'})
            &&\text{for $(\symM \of \plug{\BB}{e \equiv e' \bye \Box}) \in \Phi$}
          \end{aligned}
        }
      }{
        \Theta; \Gamma \types u \cmpE v : A
      }

      \infer{
        \notmatches{\Xi}{P}{A}
        \quad\text{for $(\rawRule{\Xi, \sym{s} \of (\Box : P), \sym{t} \of (\Box : P), \Phi}{\sym{s} \equiv \sym{t} : P}) \in \mathcal{E}$}
        \\\\
        \Theta; \Gamma \types u \compute u' : A \\
        \Theta; \Gamma \types v \compute v' : A \\
        \Theta; \Gamma \types u' \cmpN v' : A \\
      }{
        \Theta; \Gamma \types u \cmpE v : A
      }

      \infer{
      }{
        \Theta; \Gamma \types \var a \cmpN \var a : A
      }

      \infer{
        {\begin{aligned}
        &\text{Rule for $\symS$ is $\rawRule{\symM_1 \of \BB_1, \ldots, \symM_n \of \BB_n}{\plug{\B'}{\symS(\genapp{\symM}_1, \ldots, \genapp{\symM}_n)}}$}
        \\
        &{\begin{aligned}
           \Theta; \Gamma &\types \plug{(\act{\finmap{\symM_1 \mto e_1, \ldots, \symM_{i-1} \mto e_{i-1}}} \BB_i)}{e_i \cmpN e'_i}
           &&\quad\text{if $i \in \principal{\symS}$} \\
           \Theta; \Gamma &\types \plug{(\act{\finmap{\symM_1 \mto e_1, \ldots, \symM_{i-1} \mto e_{i-1}}} \BB_i)}{e_i \cmp e'_i}
           &&\quad\text{if $i \notin \principal{\symS}$}
         \end{aligned}}
        \end{aligned}}
      }{
        \Theta; \Gamma \types \plug{\B}{\symS(\vec{e}) \cmpN \symS(\vec{e'})}
      }
    \end{mathpar}
    \end{ruleframe}
  \end{center}
  \caption{Equality checking with extensionality rules~$\mathcal{E}$ and principal arguments~$\wp$.}
  \label{fig:equality-checking}
\end{figure}

General checking $\Theta; \Gamma \types \plug{\BB}{e \cmp e'}$ descends under abstractions. It compares types by normalizing them, as there are no extensionality rules for types. Terms are compared by the type-directed phase, where the type is first normalized so that it can be matched against extensionality rules.

The type-directed phase checks $\Theta; \Gamma \types u \cmpE v : A$ by looking for an extensionality rule that matches~$A$, and applying the rule to reduce the task to verification of the equational premises of the rule. The  clause uses the notation $\plug{\BB}{e \equiv e' \bye \Box}$, which turns an object boundary into an equation boundary, as follows:
\begin{align*}
  \plug{(\Box : A)}{s \equiv t \bye \Box}
  &= (s \equiv t : A \bye \Box),
  \\
  \plug{(\isType{\Box})}{A \equiv B \bye \Box}
  &= (A \equiv B \bye \Box),
  \\
  \plug{(\abstr{x \of A} \; \BB)}{\abstr{x} e \equiv \abstr{x} e' \bye \Box}
  &= \abstr{x \of A} (\plug{\BB}{e \equiv e' \bye \Box}).
\end{align*}
If no extensionality rule applies, the terms~$u$ and~$v$ are normalized and compared by the normalization phase.

The normalization phase compares normalized expressions~$\symS(\vec{e})$ and~$\symS(\vec{e}')$ by comparing their arguments, where the principal arguments are compared by the normalization phase because they have already been normalized, while the non-principal ones are subjected to general equality comparison.

The clauses in~\cref{fig:equality-checking} are readily turned into an equality-checking algorithm, because they are directed by the syntax of their goals. Application of extensionality rules is a possible source of non-determinism, as a type may match several extensionality rules, and also a source of non-termination, as there is no guarantee that eventually no extensionality rules will be applicable. We discuss strategies for dealing with these issues in~\cref{sec:deter-term-complet}.

%%% Local Variables:
%%% mode: latex
%%% TeX-master: "equality-checking"
%%% End:

% chktex-file 46
% chktex-file 45
% chktex-file 1
% chktex-file 26
% chktex-file 24
\subsection{Soundness of equality checking}
\label{sec:sound-algor}

In this section we prove that the normalization and equality checking algorithms are sound.
Because normalization and equality checking are intertwined, we prove
\cref{thm:normalization-sound} and \cref{thm:checking-sound} by mutual structural induction.

\begin{thm}[Soundness of normalization]
  \label{thm:normalization-sound}
  In a standard type theory, given a family $\mathcal{C}$ of computation rules, and a specification of principal arguments~$\wp$, the following hold, where $\BB$ and $\B$ are object boundaries:
  \begin{enumerate}[wide]
  \item
    \label{it:sound-compute}%
    If $\Theta; \Gamma \,{\types}\, \plug{\BB}{e}$ and $\Theta; \Gamma \,{\types}\, \plug{\BB}{e \compute e'}$
    then $\Theta; \Gamma \,{\types}\, \plug{\BB}{e \equiv e'}$ and $\Theta; \Gamma \,{\types}\, \plug{\BB}{e'}$.

  \item
    \label{it:sound-computeP}%
    If $\Theta; \Gamma \types \plug{\B}{e}$ and $\Theta; \Gamma \types \plug{\B}{e \computeP e'}$
    then $\Theta; \Gamma \types \plug{\B}{e \equiv e'}$ and $\Theta; \Gamma \types \plug{\B}{e'}$.
  \item
    \label{it:sound-computeC}%
    If $\Theta; \Gamma \types \plug{\B}{e}$ and $\Theta; \Gamma \types \plug{\B}{e \computeC e'}$
    then $\Theta; \Gamma \types \plug{\B}{e \equiv e'}$ and $\Theta; \Gamma \types \plug{\B}{e'}$.
  \end{enumerate}
\end{thm}

\begin{proof}
  We establish soundness of the rules from \cref{fig:normalization} by mutual structural induction on the derivations.
  Derivability of $\Theta; \Gamma \types \plug{\BB}{e'}$ in (\ref{it:sound-compute}) and of $\Theta; \Gamma \types \plug{\B}{e'}$ in (\ref{it:sound-computeP}) and~(\ref{it:sound-computeC}) follows already from \cref{prop:presuppositivity}, but we include these nonetheless as they will be needed in \cref{thm:checking-sound}.

  \medskip\noindent\emph{Part (\ref{it:sound-compute}):}
  The case of free variables follows by reflexivity and the variable rule.

  If the derivation ends with
  \begin{equation*}
    \infer{
      \var{a} \notin |\Gamma|
      \\
      \Theta; \Gamma, \var{a} \of A \types \plug{(\BB[\var{a}/x])}{e[\var{a}/x] \compute e'}
    }{
      \Theta; \Gamma \types \abstr{x \of A} \; \plug{\BB}{\abstr{x} e \compute \abstr{x} e'[x/\var{a}]}
    }
  \end{equation*}
  then by induction hypothesis
  \begin{align*}
    \Theta; \Gamma, \var{a} \of A &\types \plug{(\BB[\var{a}/x])}{e'}, \\
    \Theta; \Gamma, \var{a} \of A &\types \plug{(\BB[\var{a}/x])}{e[\var{a}/x] \equiv e'}.
  \end{align*}
  We may apply \rref{TT-Abstr} to these, because $\Theta; \Gamma \types \isType{A}$
  holds by inversion on the assumption $\Theta; \Gamma \types \abstr{x \of A} \; \plug{\BB}{\abstr{x}e}$.

  If the derivation ends with
  \begin{equation*}
    \infer{
      \Theta; \Gamma \types \plug{\B}{\symS(\vec{e}) \computeP \symS(\vec{e}')}
      \\\\
      \notmatches{\Xi}{p}{\symS(\vec{e}')}
      \quad\text{for $(\rawRule{\Xi}{\plug{\B'}{p \equiv v}}) \in \mathcal{C}$}
    }{
      \Theta; \Gamma \types \plug{\B}{\symS(\vec{e}) \compute \symS(\vec{e}')}
    }
  \end{equation*}
  then the claim follows by the induction hypothesis (\ref{it:sound-computeP}) for the first premise.

  The remaining case is
  \begin{equation*}
    \infer{
      \Theta; \Gamma \types \plug{\B}{\symS(\vec{e}) \computeP \symS(\vec{e}')} \\
      \Theta; \Gamma \types \plug{\B}{\symS(\vec{e'}) \computeC e''} \\
      \Theta; \Gamma \types \plug{\B}{e'' \compute e'''}
    }{
      \Theta; \Gamma \types \plug{\B}{\symS(\vec{e}) \compute e'''}
    }
  \end{equation*}
  The induction hypothesis for the last premise secures $\Theta; \Gamma \types \plug{\B}{e'''}$,
  while the induction hypotheses for all three premises yield
  \begin{equation*}
    \Theta; \Gamma \types \plug{\B}{\symS(\vec{e}) \equiv \symS(\vec{e}')}, \qquad
    \Theta; \Gamma \types \plug{\B}{\symS(\vec{e'}) \equiv e''}, \qquad
    \Theta; \Gamma \types \plug{\B}{e'' \equiv e'''}.
  \end{equation*}
  We may string these together using transitivity to derive $\Theta; \Gamma \types \plug{\B}{\symS(\vec{e}) \equiv e'''}$.

  \medskip\noindent\emph{Part (\ref{it:sound-computeP}):}
  Suppose the rule for $\symS$ is
  \begin{equation*}
    \rawRule
      {\symM_1 \of \BB_1, \ldots, \symM_n \of \BB_n}{\plug{\B'}
      {\symS(\genapp{\symM}_1, \ldots, \genapp{\symM}_n)}},
  \end{equation*}
  and consider normalization of principal arguments
  \begin{equation*}
    \infer{
      {\begin{aligned}
         \Theta; \Gamma \types \plug{(\upact{I}{i} \BB_i)}{e_i \compute e'_i}
         &\quad\text{if $i \in \principal{\symS}$} \\
         e_i = e'_i
         &\quad\text{if $i \notin \principal{\symS}$}
       \end{aligned}}
    }{
      \Theta; \Gamma \types \plug{\B}{\symS(\vec{e}) \computeP \symS(\vec{e}')}
    }
  \end{equation*}
  where $I = \finmap{\symM_1 \mto e_1, \ldots, \symM_n \mto e_n}$.
  For $i = 1, \ldots, n$, we have
  \begin{equation*}
    \Theta; \Gamma \types \plug{(\upact{I}{i} \BB_i)}{e_i \equiv e'_i}
    \qquad\text{and}\qquad
    \Theta; \Gamma \types \plug{(\upact{I}{i} \BB_i)}{e_i}.
  \end{equation*}
  Indeed, for $i \in \principal{\symS}$ the above are just the induction hypotheses of a premise, while for
  $i \notin \principal{\symS}$ they respectively hold by reflexivity and an application of
  \cref{cor:args-derivable} to $\Theta; \Gamma \types \plug{\B}{\symS(\vec{e})}$.
  Therefore, the instantiation $J = \finmap{\symM_1 \mto e_1', \ldots, \symM_n \mto e_n'}$ is judgementally equal to~$I$, and because~$I$ is derivable, $J$ is derivable by \cref{lem:instantiations-equal-derivable}. From these facts
  we conclude
  \begin{align*}
    \Theta; \Gamma &\types \plug{(\act{I} \B')}{\symS(\vec{e}) \equiv \symS(\vec{e}')}
    &&\text{by the congruence rule for~$S$},\\
    \Theta; \Gamma &\types \plug{(\act{J} \B')}{\symS(\vec{e}')}
    &&\text{by the rule for~$S$}.
  \end{align*}
  If $\B' = (\isType{\Box})$, we are done. If $\B' = (\Box : A)$ and $\B = (\Box : B)$ then we derive $\Theta; \Gamma \types \act{I} A \equiv \act{J} A$ by \cref{thm:eq-inst-admit} and $\Theta ; \Gamma \types \act{I} A \equiv B$ by \cref{thm:uniqueness-of-typing} on $\Theta; \Gamma \types \plug{\B}{\symS(\vec{e})}$ and convert the judgements along them.

  \medskip\noindent\emph{Part (\ref{it:sound-computeC}):}
  Consider an application of a type computation rule
  \begin{equation*}
    \infer{
      (\rawRule{\Xi}{P \equiv B}) \in \mathcal{C}
      \\
      \matches{\Xi}{P}{A}{I} \\\\
      \Theta; \Gamma \types \act{I} (\plug{\BB}{e \cmp e'})
      \quad\text{for $(\symM \of \plug{\BB}{e \equiv e' \bye \Box}) \in \Xi$}
    }{
      \Theta; \Gamma \types A \computeC \act{I} B
    }
  \end{equation*}
  \Cref{thm:checking-sound} ensures $\Theta; \Gamma \types \act{I} (\plug{\BB}{e \equiv e'})$ for every $(\symM \of \plug{\BB}{e \equiv e' \bye \Box}) \in \Xi$.
  Therefore, since $\rawRule{\Xi}{\isType{P}}$ is object-invertible and $\Theta; \Gamma \types \isType{\act{I} P}$ has been assumed (note that $\act{I} P = A$), it follows by~\cref{prop:object-invertible-promote} that $I$ is derivable. We now instantiate the computation rule $\rawRule{\Xi}{P \equiv B}$ by $I$ to get
  $\Theta; \Gamma \types A \equiv \act{I} B$ and appeal to \cref{prop:presuppositivity} for $\Theta; \Gamma \types \isType{\act{I} B}$.

  It remains to establish the soundness of a derivation ending with a term computation rule
  \begin{equation*}
    \infer{
      (\rawRule{\Xi}{p \equiv v : B}) \in \mathcal{C}
      \\
      \matches{\Xi}{p}{s}{I}
      \\\\
      \Theta; \Gamma \types \act{I} (\plug{\BB}{e \cmp e'})
      \quad\text{for $(\symM \of \plug{\BB}{e \equiv e' \bye \Box}) \in \Xi$}
    }{
      \Theta; \Gamma \types s \computeC \act{I} v : A
    }
  \end{equation*}
  \Cref{thm:checking-sound} ensures $\Theta; \Gamma \types \act{I} (\plug{\BB}{e \equiv e'})$ for every $(\symM \of \plug{\BB}{e \equiv e' \bye \Box}) \in \Xi$.
  Observe that since by~\cref{def:term-computation-rule} $p$ is a term symbol application, $\mv{\natty{\Xi; \emptyCtx}{p}} \subseteq \mv{p}$ and  $\act{I} p = s$ imply $\act{I} (\natty{\Xi; \emptyCtx}{p}) = \natty{\Theta; \Gamma}{s}$ by~\cref{prop:natty-inst}.
  Because $\Theta; \Gamma \types s : A$ is derivable, so is $\Theta; \Gamma \types s : \natty{\Theta; \Gamma}{s}$ by \cref{cor:natty-derivable}, which equals $\Theta; \Gamma \types \act{I} p : \act{I} (\natty{\Theta; \Gamma}{p})$.
  We may apply~\cref{prop:object-invertible-promote} to the object-invertible rule $\rawRule{\Xi}{p : \natty{\Theta; \Gamma}{p}}$ to establish that $I$ is derivable. By instantiating the computation rule $\rawRule{\Xi}{p \equiv v : B}$  with $I$ we obtain
  \begin{equation*}
    \Theta; \Gamma \types s \equiv \act{I} v : \act{I} B
  \end{equation*}
  and convert it along $\Theta; \Gamma \types \act{I} B \equiv A$ to the desired form, because \cref{prop:presuppositivity} implies $\Theta; \Gamma \types s : \act{I} B$ and \cref{thm:uniqueness-of-typing} that $\Theta; \Gamma \types \act{I} B \equiv A$.
  The last claim follows once again from \cref{prop:presuppositivity}.
\end{proof}

\begin{thm}[Soundness of equality checking]
  \label{thm:checking-sound}%
  In a standard type theory, given families $\mathcal{C}$ and $\mathcal{E}$ of computation and extensionality rules, and a specification of principal arguments~$\wp$, the following hold, where $\BB$ is an object boundary:
  \begin{enumerate}
  \item
    \label{it:sound-cmp}%
    $\Theta; \Gamma \types \plug{\BB}{e \equiv e'}$ holds if
    \begin{equation*}
      \Theta; \Gamma \types \plug{\BB}{e},
      \qquad
      \Theta; \Gamma \types \plug{\BB}{e'},
      \quad\text{and}\quad
      \Theta; \Gamma \types \plug{\BB}{e \cmp e'}.
    \end{equation*}

  \item
    \label{it:sound-cmpE}%
    $\Theta; \Gamma \types u \equiv v : A$ holds if
    \begin{equation*}
    \Theta; \Gamma \types u : A,
    \qquad
    \Theta; \Gamma \types v : A,
    \quad\text{and}\quad
    \Theta; \Gamma \types u \cmpE v : A.
    \end{equation*}

  \item
    \label{it:sound-cmpN-ty}%
    $\Theta; \Gamma \types A \equiv B$ holds if
    \begin{equation*}
      \Theta; \Gamma \types \isType{A},
      \qquad
      \Theta; \Gamma \types \isType{B},
      \quad\text{and}\quad
      \Theta; \Gamma \types A \cmpN B.
    \end{equation*}

  \item
    \label{it:sound-cmpN-tm}%
    $\Theta; \Gamma \types u \equiv v : A$ holds if
    \begin{equation*}
      \Theta; \Gamma \types u : A,
      \qquad
      \Theta; \Gamma \types v : A,
      \quad\text{and}\quad
      \Theta; \Gamma \types u \cmpN v : A.
    \end{equation*}

  \end{enumerate}
\end{thm}

\begin{proof}
  We proceed by mutual structural induction on the derivation.

  \medskip\noindent\emph{Part (\ref{it:sound-cmp}):}
  Consider a derivation ending with an abstraction
  \begin{equation*}
    \infer{
      \var{a} \notin |\Gamma| \\
      \Theta; \Gamma, \var{a} \of A \types \plug{(\BB[\var{a}/x])}{e[\var{a}/x] \sim e'[\var{a}/x]}
    }{
      \Theta; \Gamma \types \abstr{x \of A} \; \plug{\BB}{\abstr{x}e \sim \abstr{x}e'}
    }
  \end{equation*}
  By inverting the assumptions we get
  \begin{equation*}
    \Theta; \Gamma, \var{a} \of A \types \plug{\BB[\var{a}/x]}{e[\var{a}/x]}
    \qquad\text{and}\qquad
    \Theta; \Gamma, \var{a} \of A \types \plug{\BB[\var{a}/x]}{e'[\var{a}/x]},
  \end{equation*}
  as well as $\Theta; \Gamma \types \abstr{x \of A} \; \plug{\BB}{\abstr{x}e}$.
  Now the induction hypothesis for the premise yields
  \begin{equation*}
    \Theta; \Gamma, \var{a} \of A \types \plug{(\BB[\var{a}/x])}{e[\var{a}/x] \equiv e'[\var{a}/x]},
  \end{equation*}
  which we may abstract with \rref{TT-Abstr}.

  If the derivation ends with
  \begin{equation*}
    \infer{
      \Theta; \Gamma \types \isType{(A \compute A')} \\
      \Theta; \Gamma \types \isType{(B \compute B')} \\
      \Theta; \Gamma \types \isType{(A' \cmpN B')}
    }{
      \Theta; \Gamma \types \isType{(A \cmp B)}
    }
  \end{equation*}
  then \cref{thm:normalization-sound} applied to the first two premises gives
  \begin{align*}
    \Theta; \Gamma &\types A \equiv A', &
    \Theta; \Gamma &\types \isType{A'}, \\
    \Theta; \Gamma &\types B \equiv B', &
    \Theta; \Gamma &\types \isType{B'},
  \end{align*}
  and then the induction hypothesis for the last premise
  $\Theta; \Gamma \types A' \equiv B'$.
  From these we may derive $\Theta; \Gamma \types A \equiv B$ easily.

  Suppose the derivation ends with
  \begin{equation*}
    \infer{
      \Theta; \Gamma \types \isType{(A \compute A')} \\
      \Theta; \Gamma \types u \cmpE v : A'
    }{
      \Theta; \Gamma \types u \cmp v : A
    }
  \end{equation*}
  By \cref{prop:presuppositivity} applied to the assumption we see that $\Theta; \Gamma \types \isType{A}$,
  hence we may apply \cref{thm:normalization-sound} to the first premise and get
  \begin{equation*}
    \Theta; \Gamma \types A \equiv A'
    \qquad\text{\and}\qquad
    \Theta; \Gamma \types \isType{A'}
  \end{equation*}
  We convert the assumptions along the above equation to
  \begin{equation*}
    \Theta; \Gamma \types u : A'
    \qquad\text{and}\qquad
    \Theta; \Gamma \types v : A'
  \end{equation*}
  so that we may apply the induction hypothesis to the second premise and obtain
  $\Theta; \Gamma \types u \equiv v : A'$.
  One more conversion is then needed to derive $\Theta; \Gamma \types u \equiv v : A$.

  \medskip\noindent\emph{Part (\ref{it:sound-cmpE}):}
  If the derivation ends with
  \begin{equation*}
    \infer{
      \notmatches{\Xi}{A}{P}
      \quad\text{for $(\rawRule{\Xi, \sym{s} \of (\Box : P), \sym{t} \of (\Box : P), \Phi}{\sym{s} \equiv \sym{t} : P}) \in \mathcal{E}$}
      \\\\
      \Theta; \Gamma \types u \compute u' : A \\
      \Theta; \Gamma \types v \compute v' : A \\
      \Theta; \Gamma \types u' \cmpN v' : A \\
    }{
      \Theta; \Gamma \types u \cmpE v : A
    }
  \end{equation*}
  then \cref{thm:normalization-sound} applied to the first two premises establishes
  \begin{align*}
    \Theta; \Gamma &\types u \equiv u' : A &
    \Theta; \Gamma &\types u' : A \\
    \Theta; \Gamma &\types v \equiv v' : A &
    \Theta; \Gamma &\types v' : A
  \end{align*}
  Then the induction hypothesis tells us that
  $\Theta; \Gamma \types u' \equiv v' : A$.
  It is now easy to combine the derived equalities into
  $\Theta; \Gamma \types u \equiv v : A$.

  If the derivation ends with an application of an extensionality rule
  \begin{equation*}
    \infer{
      (\rawRule{\Xi, \sym{s} \of (\Box : P), \sym{t} \of (\Box : P), \Phi}{\sym{s} \equiv \sym{t} : P})
      \in \mathcal{E}
      \\
      \matches{\Xi}{A}{P}{I}
      \\\\
      {
        \begin{aligned}
          \Theta; \Gamma &\types \act{I} (\plug{\BB}{e \cmp e'})
          &&\text{for $(\symM \of \plug{\BB}{e \equiv e' \bye \Box}) \in \Xi$}
          \\
          \Theta; \Gamma &\types \act{\finmap{I, \sym{s} \mto u, \sym{t} \mto v}} (\plug{\BB}{e \cmp e'})
          &&\text{for $(\symM \of \plug{\BB}{e \equiv e' \bye \Box}) \in \Phi$}
        \end{aligned}
      }
    }{
      \Theta; \Gamma \types u \cmpE v : A
    }
  \end{equation*}
  then $\Theta; \Gamma \types \isType{A}$ follows from $\Theta; \Gamma \types u : A$ by \cref{prop:presuppositivity}.
  Induction hypotheses for the premises give
  \begin{align}
    \label{eq:soundness-ext-eqpremises}
    \Theta; \Gamma &\types \act{I} (\plug{\BB}{e \equiv e'})
    &&\text{for $(\symM \of \plug{\BB}{e \equiv e' \bye \Box}) \in \Xi$}
    \\
    \label{eq:soundness-ext-eqpremises-1}
    \Theta; \Gamma &\types \act{\finmap{I, \sym{s} \mto u, \sym{t} \mto v}} (\plug{\BB}{e \equiv e'})
    &&\text{for $(\symM \of \plug{\BB}{e \equiv e' \bye \Box}) \in \Phi$}
  \end{align}
  Because $\rawRule{\Xi}{\isType{P}}$ is object-invertible, and $\act{I} P = A$ and $\Theta ; \Gamma \types \isType{A}$ is derivable, by~\cref{prop:object-invertible-promote} the instantiation $I$ is derivable too.
  We extend $I$ to the instantiation
  \begin{equation*}
    J = \finmap{I, \sym{s} \mto u, \sym{t} \mto v, \Phi \mto \dummy}
  \end{equation*}
  of the premises of the extensionality rule over $\Theta; \Gamma$, where $\Phi \mto \dummy$ signifies that the metavariables of $\Phi$ are instantiated with (suitably abstracted) dummy values.
  We claim that $J$ is derivable:
  we already know that~$I$ is derivable; derivability at $\sym{s}$ and $\sym{t}$ reduces to the assumptions $\Theta; \Gamma \types u : A$ and $\Theta; \Gamma \types u : A$; and derivability at~$\Phi$ holds by the induction hypotheses~\eqref{eq:soundness-ext-eqpremises-1}.
  When we instantiate the extensionality rule with~$J$, we obtain the desired equation.

  \medskip\noindent\emph{Parts (\ref{it:sound-cmpN-ty}) and (\ref{it:sound-cmpN-tm}):}
  The variable case $\Theta; \Gamma \types \var a \cmpN \var a : A$ is trivial.

  Suppose the rule for symbol $\symS$ is
  \begin{equation*}
    \rawRule
      {\symM_1 \of \BB_1, \ldots, \symM_n \of \BB_n}
      {\plug{\B'}{\symS(\genapp{\symM}_1, \ldots, \genapp{\symM}_n)}}
  \end{equation*}
  and the derivation ends with
  \begin{equation*}
    \infer{
      {\begin{aligned}
        \Theta; \Gamma &\types \plug{(\upact{I}{i} \BB_i)}{e_i \cmpN e'_i}
        &&\quad\text{if $i \in \principal{\symS}$} \\
        \Theta; \Gamma &\types \plug{(\upact{I}{i} \BB_i)}{e_i \cmp e'_i}
        &&\quad\text{if $i \notin \principal{\symS}$}
      \end{aligned}}
    }{
      \Theta; \Gamma \types \plug{\B}{\symS(\vec{e}) \cmpN \symS(\vec{e'})}
    }
  \end{equation*}
  where $I = \finmap{\symM_1 \mto e_1, \ldots, \symM_n \mto e_n}$, and define
  $J = \finmap{\symM_1 \mto e'_1, \ldots, \symM_n \mto e'_n}$.
  We first derive
  \begin{equation}
    \label{eq:sound-cmp-1}
    \Theta; \Gamma \types \plug{(\act{I} \B')}{\symS(\vec{e}) \equiv \symS(\vec{e'})}
  \end{equation}
  by the congruence rule associated with~$\symS$, whose premises are derived as follows:
  \begin{enumerate}

  \item
    For each $i = 1, \ldots, n$ the premise $\Theta; \Gamma \types \plug{(\upact{I}{i} \BB_i)}{e_i}$ is derivable by \cref{cor:args-derivable} applied to $\Theta; \Gamma \types \plug{\B}{\symS(\vec{e})}$.
    This also shows that~$I$ is derivable.

  \item
    For each $i = 1, \ldots, n$ such that $\BB_i$ is an object boundary, the premise $\Theta; \Gamma \types \plug{(\upact{I}{i} \BB_i)}{e_i \equiv e'_i}$ is one of the induction hypotheses.
    This also shows that $I$ and $J$ are judgementally equal, therefore $J$ is derivable by \cref{lem:instantiations-equal-derivable}.

  \item For each $i = 1, \ldots, n$ the premise $\Theta; \Gamma \types \plug{(\upact{J}{i} \BB_i)}{e'_i}$ is derivable because~$J$ is derivable.
  \end{enumerate}
  If $\B = \isType{\Box}$, we are done.
  If $\B' = (\Box : A)$ and $\B = (\Box : B)$, we convert~\eqref{eq:sound-cmp-1} along $\Theta; \Gamma \types \act{I} A \equiv B$. The equation holds by \cref{thm:uniqueness-of-typing} applied to $\Theta; \Gamma \types \symS(\vec{e}) : B$ and $\Theta; \Gamma \types \symS(\vec{e}) : \act{I} A$, where the latter is derived by \cref{prop:presuppositivity} and the former by the rule for~$\symS$.
\end{proof}

%%% Local Variables:
%%% mode: latex
%%% TeX-master: "equality-checking"
%%% End:

%%% Local Variables:
%%% mode: latex
%%% TeX-master: "equality-checking"
%%% End:
 % description of the algorithm, soundness
% chktex-file 46
% chktex-file 45
% chktex-file 1
% chktex-file 26
% chktex-file 24
\section{Discussion}
\label{sec:discussion}

The relations defined by the inductive clauses from \cref{fig:normalization,fig:equality-checking} serve as the basis of an equality checking algorithm. In order to obtain a working and useful implementation, we need to address several issues.

\subsection{Classification of rules and principal arguments}
\label{sec:choose-determ-princ}

An experienced designer of type theories is quite able to recognize computation and extensionality rules, and stitch them together by picking correct principal arguments. There is no need for such manual work, because
\cref{prop:comp-rule-suff-cond,prop:extensionality-rule-criterion} provide easily verifiable syntactic criteria for recognizing computation and extensionality rules.
The principal arguments must be chosen correctly, lest the equality checking procedure fail unnecessarily or enter an infinite loop, as shown by the following example.

\begin{exa}
\label{exa:principal-arguments}%
Consider the computation and extensionality rules for simple products shown in \cref{fig:simple-product-equality-rules}, where we ignore the linearity requirements, as they just obscure the point we wish to make.
%
% NO EMPTY LINE HERE
%
\begin{figure}[htbp]
  \centering
\small
\begin{ruleframe}
\begin{mathpar}
  \infer{
    \types \isType{\sym{A}}
    \\
    \types \isType{\sym{B}}
    \\
    \types \sym{s} : \sym{A}
    \\
    \types \sym{t} : \sym{B}
  }{
    \types
    \sym{fst}(\sym{A}, \sym{B}, \sym{pair} (\sym{A}, \sym{B}, \sym{s}, \sym{t}))
    \equiv
    \sym{s}
    : \sym{A}
  }

  \infer{
    \types \isType{\sym{A}}
    \\
    \types \isType{\sym{B}}
    \\
    \types \sym{s} : \sym{A}
    \\
    \types \sym{t} : \sym{B}
  }{
    \types
    \sym{snd}(\sym{A}, \sym{B}, \sym{pair} (\sym{A}, \sym{B}, \sym{s}, \sym{t}))
    \equiv
    \sym{t}
    : \sym{B}
  }

  \infer{
    \types \isType{\sym{A}}
    \\
    \types \isType{\sym{B}}
    \\
    \types \sym{s} : \sym{A} \times \sym{B}
    \\
    \types \sym{t} : \sym{A} \times \sym{B}
    \\
    \types \mathsf{fst}(\sym{A}, \sym{B}, \sym{s}) \equiv \mathsf{fst}(\sym{A}, \sym{B}, \sym{t}) : \sym{A}
    \\
    \types \mathsf{snd}(\sym{A}, \sym{B}, \sym{s}) \equiv \mathsf{snd}(\sym{A}, \sym{B}, \sym{t}) : \sym{B}
  }{
    \types \sym{s} \equiv \sym{t} : \sym{A} \times \sym{B}
  }
\end{mathpar}
\end{ruleframe}
  \caption{Computation and extensionality rules for simple products}
  \label{fig:simple-product-equality-rules}
\end{figure}
%
% NO EMPTY LINES HERE
%
Without any principal arguments, the algorithm fails to apply the first computation rule to $\sym{fst}(A, B, u)$
in case $u$ normalizes to a pair.
More ominous is the infinite loop that is entered on checking
\begin{equation*}
  \emptyExt; x \of A \times B , y \of A \times B \types x \equiv y : A \times B,
\end{equation*}
where we assume that $A$ and $B$ are already normalized.
The algorithm performs the following steps (where all judgements are placed in the variable context $\emptyExt; x \of A \times B, y \of A \times B$).
First, the extensionality phase reduces the equation to
\begin{equation*}
  \sym{fst}(A, B, x) \equiv \sym{fst}(A, B, y) : A,
  \qquad
  \sym{snd}(A, B, x) \equiv \sym{snd}(A, B, y) : B.
\end{equation*}
after which the normalization verifies the first equation by comparing
\begin{equation*}
  A \equiv A,
  \qquad
  B \equiv B,
  \qquad
  x \equiv y : A \times B.
\end{equation*}
We may short-circuit the first two equalities, but checking the third one leads back to the original one, \emph{unless} the third argument of $\sym{fst}$ is principal, in which case the algorithm persists in the normalization phase and fails immediately, as it should.
\end{exa}

The previous example suggests that we can read off the principal arguments either from extensionality rules, by looking for occurrences of the left and right-hand sides in the subsidiary equalities, or from computation rules, by inspecting the syntactic form of the left-hand side of the rule.
We have analyzed a number of standard computation and extensionality rules and identified the following strategy for automatic determination of principal arguments, which we also implemented:
\begin{quote}
  \emph{
    The $i$-th argument of~$\symS$ is principal if there is a computation rule $\rawRule{\Xi}{\plug{\B}{p \equiv v}}$ such that~$\symS(e_1, \ldots, e_n)$ appears as a sub-pattern of~$p$ and~$e_i$ is neither of the form $\symM()$ nor $\abstr{\vec{x}} \symM(\vec{x})$.
  }
\end{quote}
In many cases, among others the simply-typed $\lambda$-calculus, inductive types, and intensional Martin-Löf type theory, the strategy leads to weak head-normal forms. We postpone the pursuit of deeper understanding of this phenomenon to another time.

It would be interesting to combine principal arguments with another common technique for controlling applications of extensionality rules, namely \emph{neutral forms}. Roughly, the principal arguments would still tell which arguments are normalized, but not whether they are compared structurally. Instead, we always compare them recursively, but skip the type-directed phase in $s \equiv t : A$ when the syntactic forms of~$s$ and~$t$ are neutral, i.e., they indicate that application of extensionality rules cannot lead to further progress. For instance, when checking $x \equiv y : A \times B$ in \cref{exa:principal-arguments}, there is no benefit to applying projections to the variables~$x$ and~$y$. Each specific type theory has its own neutral forms, if any, so the user would have to describe these. In some cases it might even be possible to detect the neutral forms automatically.

\subsection{Determinism, termination and completeness}
\label{sec:deter-term-complet}

The inductive clauses in \cref{fig:normalization,fig:equality-checking} could be implemented either as proof search, or as a streamlined algorithm based on normalization. Proof assistants typically implement the latter strategy, because they work with type theories whose normalization is confluent and terminating, and equality checking requires no backtracking. We use the same strategy, so we ought to address non-determinism and non-termination.

A computation or extensionality rule cannot be the source of non-determinism on its own, because \cref{def:type-computation-rule,def:term-computation-rule,def:extensionality-rule} prescribe determinism. However, in either phase of the algorithm several rules may be applicable at the same time, which leads to non-determinism, and we saw in \cref{exa:principal-arguments} that a poor choice of principal arguments causes non-termination.
This is all quite familiar, and so are techniques for ensuring that all is well, including confluence checking and termination arguments based on well-founded relations.
While these are doubtlessly important issues, we are not addressing them because they are independent of the algorithm itself. Instead, we aim to provide equality checking that favors generality and extensibility, while still providing soundness through \Cref{thm:checking-sound,thm:normalization-sound}. In this regard we are in good company, as recent version of Agda allow potentially unsafe user-defined computation rules, a point further discussed in~\cref{sec:related-work}.

A related question is completeness of equality checking, i.e., does the algorithm succeed in checking every derivable equation? Once again, our position is the same: completeness is important, both theoretically and from a practical point of view, but is not the topic of the present paper. Numerous techniques for establishing completeness of equality checking are known, and these can be applied to any specific instantiation of our algorithm. An interesting direction to pursue would be adaptation of such techniques to our general setting.

%%% Local Variables:
%%% mode: latex
%%% TeX-master: "equality-checking"
%%% End:
 % discussion about principal arguments, completeness etc.
% chktex-file 46
% chktex-file 45
% chktex-file 1
% chktex-file 26
% chktex-file 24
\section{Implementation}
\label{sec:implementation}

Having laid out the algorithm, we report on our experience with its implementation in the Andromeda~2 proof assistant~\cite{andromeda-1,andromeda-site,ICMS2020-equality-checking}, in which the user may define any work in any standard type theory.
It is an LCF style proof assistant, i.e., a meta-level programming language with abstract datatypes of judgements, boundaries, and derived rules whose construction and application is controlled by a trusted nucleus (consisting of around 4200 lines of OCaml code).

The nucleus implements \defemph{context-free type theory}, a variant of type theory in which there are no contexts. Instead, each free variable is tagged with its type and each metavariable with its boundary, as explained in~\cite{bauer:_finit}. Since there are no contexts, a mechanism is needed for tracking proof-irrelevant uses of metavariables and variables, which may occur in derivations of equalities. For this purpose, equality judgements take the form
\begin{equation*}
  A \equiv B \bye \alpha
  \qquad\text{and}\qquad
  s \equiv t : A \bye \alpha
\end{equation*}
where~$\alpha$ is an \emph{assumption set} whose elements are those metavariables and variables that are used to derive the equality but do not appear in its boundary. The assumptions sets are also recorded in term conversions. As far as the equality checking algorithm is concerned, this is an annoying but inessential complication, because all conversions must be performed explicitly and carefully accounted for.

The implementation of the equality checking algorithm comprises around 1400 lines of OCaml code which reside outside of the trusted nucleus, so that each reasoning step must be passed to the nucleus for validation. The overhead of such a policy is significant, but worth paying in exchange for keeping the nucleus small and uncorrupted, at least in the initial, experimental phase.

Our rudimentary implementation is quite inefficient and cannot compete with the equality checkers found in mature proof assistants. The interesting question is not whether we could try harder to significantly speed up the algorithm, which presumably we could, but whether the design of the algorithm makes it inherently inefficient. We argue that this is not the case.
First, we may trade safety for efficiency by placing equality checking into the trusted nucleus, as many proof assistants do, so that we need not check every single step of the algorithm.
Second, even though term equality is typed, the normalization procedure is essentially untyped. Indeed, when the rules in \cref{fig:normalization} are used to normalize $\Theta; \Gamma \types t : A$ they never modify~$A$, and only ever inspect~$t$, which allows us to ignore~$A$ while rewriting~$t$. The soundness of the algorithm guarantees that the normalized term will still have type~$A$.

%%% Local Variables:
%%% mode: latex
%%% TeX-master: "equality-checking"
%%% End:
 % Andromeda~2implementation (including context-free version)
% chktex-file 46
% chktex-file 45
% chktex-file 1
% chktex-file 26
% chktex-file 24
\section{Examples}
\label{sec:examples-code}

\begin{figure}[htbp]
  \centering
\begin{ruleframe}
\begin{lstlisting}
require eq ;;

rule Π (A type) ({x : A} B type) type ;;
rule λ (A type) ({x : A} B type) ({x : A} e : B{x}) : Π A B ;;
rule app (A type) ({x : A} B type) (s : Π A B) (a : A) : B{a} ;;

rule Π_β
  (A type) ({x:A} B type)
  ({x : A} s : B{x}) (t : A) :
  app A B (λ A B s) t ≡ s{t} : B{t} ;;

rule sym_ty (A type) (B type) (A ≡ B) : B ≡ A ;;

rule Π_β_linear
   (A₁ type) ({x:A₁} B₁ type)
   (A₂ type) ({x:A₂} B₂ type)
   ({x:A₂} s : B₂{x}) (t : A₁)
   (A₂ ≡ A₁ by ξ) ({x : A₂} B₂{x} ≡ B₁{convert x ξ} by ζ)
   : app A₁ B₁ (convert (λ A₂ B₂ s)
     (congruence (Π A₂ B₂) (Π A₁ B₁) ξ ζ)) t
     ≡ convert s{convert t (sym_ty A₂ A₁ ξ)}
               ζ{convert t (sym_ty A₂ A₁ ξ)} : B₁{t} ;;

eq.add_rule Π_β_linear ;;

rule Π_ext (A type) ({x : A} B type)
  (f : Π A B) (g : Π A B)
  ({x : A} app A B f x ≡ app A B g x : B{x})
  : f ≡ g : Π A B;;

eq.add_rule Π_ext;;
\end{lstlisting}
\end{ruleframe}
  \caption{Dependent products in Andromeda~2.}
  \label{fig:dependent-products}
\end{figure}

The example in \Cref{fig:dependent-products} shows how dependent products are formalized in Andromeda~2. The rules are direct transcriptions of the usual ones. We linearize the $\beta$-rule as shown in~\cref{exa:linearize} to make it a computation rule. We do so by explicitly converting \lstinline|λ A₂ B₂ s| along the equality \lstinline|Π A₂ B₂ ≡ Π A₁ B₁|, which holds by a congruence rule and the premises \lstinline|ξ| and \lstinline|ζ|.

The calls to \lstinline|eq.add_rule| pass equality rules to the equality checking algorithm, which employs \cref{prop:comp-rule-suff-cond,prop:extensionality-rule-criterion} to automatically classify the inputs as computation or extensionality rules. It also determines which arguments are principal by using the technique from~\cref{sec:choose-determ-princ}. In the example shown, the linearized rule \lstinline|Π_β_linear| is classified by the algorithm as computation rule, \lstinline|Π_ext| as extensionality rule, and the the third argument of \lstinline|app| is declared principal.

Many a newcomer to Martin-Löf type theory is disappointed to learn that only one of equalities $0 + n = n$ and $n + 0 = n$ holds judgementally. In fact, there is strong temptation to pass to extensional type theory just so that a more symmetric notion of equality is recovered, but then one has to give up decidable equality checking. The example in \cref{fig:addition-nat,fig:equality-type} shows how our algorithm combines the best of both worlds and demonstrates further capabilities of the implementation.

First, \Cref{fig:equality-type} shows a formalization of extensional equality types, whose distinguishing feature is the equality reflection principle~\lstinline|equality_reflection|, which states that the equality type \lstinline|Eq| reflects into judgemental equality. Instead of postulating the familiar eliminator~\lstinline|J|, it is more convenient to use an equivalent formulation that uses the judgemental uniqueness of equality proofs \lstinline|uip|, see \cref{ex:ext-rules-no-equations}. Note that \lstinline|uip| is installed as an extensionality rule into the equality checker. It is well known that equality reflection makes equality checking undecidable, so the equality checker will not be able to prove all equalities. Nevertheless, we expect it to be still quite useful and well behaved.

\begin{figure}[htbp]
  \centering
\begin{ruleframe}
\begin{lstlisting}
require eq ;;

rule Eq (A type) (a : A) (b : A) type ;;
rule refl (A type) (a : A) : Eq A a a ;;

rule equality_reflection
  (A type) (a : A) (b : A) (_ : Eq A a b)
  : a ≡ b : A ;;

rule uip (A type) (a : A) (b : A)
     (p : Eq A a b) (q : Eq A a b)
     : p ≡ q : Eq A a b ;;

eq.add_rule uip ;;
\end{lstlisting}
\end{ruleframe}
  \caption{Extensional equality type in Andromeda~2.}
  \label{fig:equality-type}
\end{figure}

\begin{figure}[htbp]
  \centering
\begin{ruleframe}
\begin{lstlisting}
rule ℕ type ;;
rule zero : ℕ ;;
rule succ (n : ℕ) : ℕ ;;

rule ℕ_ind
  ({_ : ℕ} C type) (x : C{zero})
  ({n : ℕ} {u : C{n}} f : C{succ n}) (n : ℕ)
  : C{n} ;;

rule ℕ_β_zero
  ({_ : ℕ} C type) (x : C{zero})
  ({n : ℕ} {u : C{n}} f : C{succ n})
  : ℕ_ind C x f zero ≡ x : C{zero} ;;

eq.add_rule ℕ_β_zero ;;

rule ℕ_β_succ
  ({_ : ℕ} C type) (x : C{zero})
  ({n : ℕ} {u : C{n}} f : C{succ n}) (n : ℕ)
  : ℕ_ind C x f (succ n) ≡ f{n, ℕ_ind C x f n} : C{succ n} ;;

eq.add_rule ℕ_β_succ ;;

rule (+) (_ : ℕ) (_ : ℕ) : ℕ ;;
rule plus_def (m : ℕ) (n : ℕ) :
  (m + n) ≡ ℕ_ind ({_} ℕ) m ({_ : ℕ} {u : ℕ} succ u) n : ℕ ;;

let plus_zero_right = derive (n : ℕ) →
  eq.add_locally plus_def
    (fun () → eq.prove ((n + zero) ≡ n : ℕ by ??)) ;;

eq.add_rule plus_zero_right ;;

let plus_succ = derive (m : ℕ) (n : ℕ) →
  eq.add_locally plus_def
    (fun () →
      eq.prove ((m + succ n) ≡ (succ (m + n)) : ℕ by ??)) ;;

eq.add_rule plus_succ ;;

let plus_zero_left = derive (k : ℕ) →
  let ap_succ = derive (m : ℕ) (n : ℕ) (p : Eq ℕ m n) →
    eq.add_locally (derive → equality_reflection ℕ m n p)
      (fun () → refl ℕ (succ m) : Eq ℕ (succ m) (succ n)) in
  eq.add_locally plus_def
    (fun () →
       equality_reflection ℕ (zero + k) k
         (ℕ_ind ({n} Eq ℕ (zero + n) n) (refl ℕ zero)
                ({n} {ih} ap_succ (zero + n) n ih) k)) ;;

eq.add_rule plus_zero_left ;;
\end{lstlisting}
\end{ruleframe}
  \caption{Addition for natural numbers in Andromeda~2.}
  \label{fig:addition-nat}
\end{figure}

We continue our example in \cref{fig:addition-nat} by postulating the natural numbers \lstinline|ℕ|. Everything up to the definition of addition is standard, where we also install the computation rules for the induction principle \lstinline|ℕ_ind| into the equality checker.
We then define addition by postulating a term symbol \lstinline|+| with the defining equality \lstinline|plus_def| which expresses addition by primitive recursion. We could use \lstinline|plus_def| as a global computation rule, but we choose to use it only \emph{locally}, with the help of the function~\lstinline|eq.add_locally|.

In the remainder of the code we prove \emph{judgemental} equalities
\begin{equation*}
  n + 0 \equiv n, \qquad
  m + \sym{succ}(n) \equiv \sym{succ}(m + n), \quad\text{and}\quad
  0 + n = n.
\end{equation*}
The first one is derived as \lstinline|plus_zero_right| using
\lstinline|plus_def| as a local computation rule together with \lstinline|eq.prove| which takes an equational boundary (where $\Box$ is written as \lstinline|??|) and runs the equality checking algorithm to generate a witness for it.
The second equality is derived as \lstinline|plus_succ| in much the same way.
The derivation of the third equality relies on equality reflection to convert a term
of the equality type \lstinline|Eq ℕ (zero + n) n| to the corresponding judgemental equality \lstinline|zero + n ≡ n : ℕ|.
We install all three equalities as computation rules.

In addition to proving equalities, we can also normalize terms with \lstinline|eq.normalize|, and compute strong normal forms (all arguments are principal) with \lstinline|eq.compute|.
In both cases we obtain not only the result, but also a certifying equality.
For example, when given \lstinline|succ zero + succ zero|, the normalizer outputs the weak head-normal form \lstinline|succ ((succ zero) + zero)|,
together with a certificate for the judgemental equality \lstinline|(succ zero) + (succ zero) ≡ succ ((succ zero) + zero) : ℕ|.
Because we installed both neutrality laws for~$0$ as computation rules, strong normalization reduces
\lstinline|(zero + x) + succ (succ zero + zero)| to \lstinline|succ (succ x) : ℕ|, where~\lstinline|x| is a free variable of type \lstinline|ℕ|.

%%% Local Variables:
%%% mode: latex
%%% TeX-master: "equality-checking"
%%% End:
 % Examples
% chktex-file 46
% chktex-file 45
% chktex-file 1
% chktex-file 26
% chktex-file 24
\section{Related work}
\label{sec:related-work}

Designing a user-extensible equality checking algorithm for type theory is a balancing act between flexibility, safety, and automation. We compare ours to that of several proof assistants that support user-extensible equality checking.

% MMT
The overall design of our algorithm is similar to the equality checking and simplification phases used in the type-reconstruction algorithm of MMT~\cite{MMT-site,MMT-paper}, a meta-meta-language for description of formal theories.
In MMT inference rules are implemented as trusted low-level executable code, which gives the system an extremely wide scope but also requires care and expertise by the user. In Andromeda~2 the user writes down the desired inference rules directly. The nucleus checks them for compliance with \cref{def:standard-type-theory} of a standard type theory before accepting them, which prevents the user from breaking the meta-theoretic properties that the nucleus relies on.

% Dedukti
Dedukti~\cite{dedukti-site} is a type-checker founded on the logical framework~$\lambda\Pi$, extended with user-defined conversion rules. Because equality in Dedukti is based on convertibility of terms, there is no support for user-defined extensionality or $\eta$-rules.
The Dedukti rewriting system supports higher-order patterns and includes a confluence checker. We see no obstacle to adding some form of confluence checking to Andromeda~2 in the future, while support for higher-order patterns would first have to overcome lack of strengthening, see the discussion following \cref{def:pattern}.

% Agda
Recent versions of the proof assistant Agda support user-definable computation rules~\cite{sprinklesAgda,cockx:TTunchained,cockx-rew}.
Like Dedukti, Agda allows higher-order patterns and provides a confluence checker.
It accepts non-linear patterns, which it linearizes and generates suitable equational premises.
In addition, it applies built-in $\eta$-rules for functions and record types during a type-directed matching phase. It seems to us that the phase could equally well use extensionality rules, which might more easily enable user-defined extensionality principles.
Agda designers point out in~\cite{cockx:TTunchained} that having \emph{local} rewrite rules would improve modularity. For example, one could parameterize code by an abstract type, together with rewrite rules it satisfies. This sort of functionality is already present in Andromeda~2, which treats all judgement forms as first-class values, so we may simply pass judgemental equalities as parameters and use them as local computation and extensionality rules.

In order to make our equality checking algorithm realistically useful, we ought to combine it with other techniques, such as existential variables, unification, and implicit arguments. Whether that can be done in full generality remains to be seen.

%%% Local Variables:
%%% mode: latex
%%% TeX-master: "equality-checking"
%%% End:
 % related work

\bibliographystyle{alphaurl}
\bibliography{equality-checking}

% \addresseshere

\end{document}